\documentclass[11pt,twoside,letter]{article}
\usepackage{amsthm, amsmath, amssymb, amsfonts}
\usepackage{mathtools}

\usepackage{bbm}
\usepackage{mathrsfs}

\usepackage[scaled]{helvet} 
\usepackage{courier} 
\normalfont

\usepackage{mathpple}

\usepackage{cancel}

\usepackage{graphicx}

\usepackage{epstopdf}

\usepackage{wrapfig}

\usepackage[shortlabels]{enumitem}

\mathtoolsset{showonlyrefs}

\usepackage[colorlinks=true, pdfstartview=FitV, linkcolor=blue, citecolor=blue, urlcolor=blue]{hyperref}
\usepackage[usenames]{color}

\usepackage{url}
\makeatletter\def\url@leostyle{%
	\@ifundefined{selectfont}{\def\UrlFont{\sf}}{\def\UrlFont{\scriptsize\ttfamily}}} \makeatother

\usepackage[margin=1.0in, letterpaper]{geometry}

\newtheorem{theorem}{Theorem}
\newtheorem{proposition}[theorem]{Proposition}
\newtheorem{lemma}[theorem]{Lemma}

\theoremstyle{definition}
\newtheorem{definition}[theorem]{Definition}
\newtheorem{corollary}[theorem]{Corollary}
\newtheorem{example}[theorem]{Example}
\theoremstyle{remark}
\newtheorem{remark}[theorem]{Remark}

\numberwithin{equation}{section}
\numberwithin{theorem}{section}

\definecolor{Red}{rgb}{0.9,0,0.0}
\definecolor{Blue}{rgb}{0,0.0,1.0}

\def\cB{\mathcal{B}}

\def\cG{\mathcal{G}}

\def\cP{\mathcal{P}}
\def\cQ{\mathcal{Q}}

\def\cT{\mathcal{T}}

\def\bE{\mathbb{E}}

\def\bN{\mathbb{N}}

\def\bP{\mathbb{P}}
\def\bQ{\mathbb{Q}}
\def\bR{\mathbb{R}}

\def\sF{\mathscr{F}}
\def\sG{\mathscr{G}}

\newcommand{\1}{\mathbbm{1}}            
\newcommand{\set}[1]{\{#1\}}            
\newcommand{\Set}[1]{\left\{#1\right\}} 
\renewcommand{\mid}{\;|\;}              
\newcommand{\Mid}{\;\Big | \;}          
\newcommand{\norm}[1]{ \| #1 \| }       

\DeclareMathOperator{\dif}{d \!}        
\DeclareMathOperator*{\esssup}{ess\,sup} 
\DeclareMathOperator*{\essinf}{ess\,inf} 

\DeclareMathOperator{\var}{\mathrm{V}@\mathrm{R}}           
\DeclareMathOperator{\wvar}{\mathrm{WV}@\mathrm{R}} 
\DeclareMathOperator{\dwvar}{\mathrm{dWV}@\mathrm{R}} 
\DeclareMathOperator{\avar}{\mathrm{AV}@\mathrm{R}}         

\usepackage{tikz}
\usetikzlibrary{matrix}

\title{Time consistency of dynamic risk measures and dynamic performance measures generated by distortion functions}

\author{
	Tomasz R. Bielecki\,\thanks{The authors acknowledge support from the National Science Foundation grant DMS-1907568.
	}\\[-0.3ex]
	\url{tbielecki@iit.edu}  \\[-0.9ex]
	\url{http://math.iit.edu/\~bielecki}
	\and
	Igor Cialenco\,\footnotemark[1] \\[-0.3ex]
	\url{cialenco@iit.edu}  \\[-0.9ex]
	\url{http://cialenco.com}
	\and
	Hao Liu\,\footnotemark[1] \\[-0.3ex]
	\url{hliu95@hawk.iit.edu}  \\[-0.9ex]	
	\and \\[-0.5em]
	{\footnotesize Department of Applied Mathematics, Illinois Institute of Technology} \\
	{\footnotesize 10 W 32nd Str, John T. Rettaliata Engineering Center, Room 220, Chicago, IL 60616, USA}\\
}

\date{ {\small  
First Circulated:  September 1, 2023  \\
This Version: September 8, 2023 }}

\begin{document}

	\maketitle
	%
	%
	
	
	{\footnotesize
		\begin{tabular}{l@{} p{350pt}}
			\hline \\[-.2em]
			\textsc{Abstract}: \ & The aim of this work is to study risk measures generated by distortion functions in a dynamic discrete time setup, and to investigate the corresponding dynamic coherent acceptability indices (DCAIs) generated by families of such risk measures. First we show that conditional version of Choquet integrals indeed are dynamic coherent risk measures (DCRMs), and also introduce the class of  dynamic  weighted value at risk measures. We prove that these two classes of risk measures  coincides. In the spirit of robust representations theorem for DCAIs, we establish some relevant properties of families of DCRMs generated by distortion functions, and then define and study the corresponding DCAIs. Second, we study the time consistency of DCRMs and  DCAIs generated by distortion functions. 	In particular, we prove that such DCRMs 	are  sub-martingale time consistent, but they are not super-martingale time consistent. We also show that DCRMs generated by distortion functions are not weakly acceptance time consistent. We also present several widely used classes of distortion functions and derive some new representations of these distortions.
			
			This manuscript is also complimented with a technical appendix, where we collect a series of results and proofs. In particular, we prove that the dynamic risk measures generated by regular distortion functions are indeed dynamic coherent risk measures. Similarly, we provide the detailed proof that the acceptability indices generated by distortion functions are dynamic coherent  acceptability indices. We also present some properties of conditional quantiles, conditional $\var$ and conditional $\avar$.

			\\[0.5em]
			\textsc{Keywords:} \ &  dynamic risk measures, dynamic acceptability indices, time consistency, distortion functions, weighted value at risk, average value at risk, low invariant risk measures, MINVAR, MAXVAR, MINMAXVAR, MAXMINVAR, Choquet integral  \\
			\textsc{MSC2010:} \ &  Primary 91B06; Secondary 91B30, 91B08    \\[1em]
			\hline
		\end{tabular}
	}

\section{Introduction}
The aim of this paper is to study various types of time consistency of a class of Dynamic Coherent Risk Measures (DCRMs) that are generated by distortion functions, as well as time consistency of the corresponding Dynamic Coherent Acceptability Indices (DCAIs).

Risk and performance measures have been staying at the core of finance and insurance industries, successfully used as tools for computing the regulatory capital requirement, in wealth management, as well as in pricing and hedging complex derivatives through a nonlinear setup. In the seminal paper by Artzner et al \cite{ArtznerDelbaenEberHeath1999}, the authors proposed a systematic approach in studying coherent risk measures (CRM), in a static (one period time) framework, as real valued functions acting on random P\&Ls, that satisfies a set of desired axioms. Namely, with $L^\infty$ denoting the space of (essentially) bounded random variables on some probability space $(\Omega, \sF, \mathbb{P})$,
a static coherent risk measure is a function $\rho: L^\infty \rightarrow [-\infty, \infty]$ that is monotone decreasing, cash-additive, sub-additive and positive homogeneous (see Section~\ref{sec:DCRM} for details). It was shown that such functions admit different representations, usually derived in the context of convex analysis and duality theory. There exists a vast literature dedicated to extensions of coherent risk measures theory to various degrees.

In the static setup, the natural pathway was to impose a  different, usually smaller, set of axioms or to considering a larger spaces on which these functions are defined. We refer the reader to \cite{DrapeauKupper2010} and references therein for a detailed overview of general theory of risk measures (convex, monetary, quasi-convex, etc).   One important class of coherent risk measures consists of the coherent risk measures generated by distortion functions, introduced in \cite{FollmerSchiedBook2004}, that take the form
\begin{equation}\label{eq:wvar-static}
	\rho^\psi(X) = - \int_{\mathbb{R}} y \dif \psi(\mathbb{P}(X \leq y)),
\end{equation}
where the distortion function $\psi:[0,1]\rightarrow [0,1]$ is non-decreasing and $\psi(0)= 0$,  $\psi(1)= 1$.  The name distortion is apparent - the function $\psi$ distorts the tails  of the distribution of the random variable $X$, which in the context of risk management corresponds to distorting  of the distribution of the large or extreme losses and gains. Clearly, if $\psi(x)=x$, then $\rho^\psi$ is the negative of the  (linear) expectation, hence no distortions. Respectively, $\rho^\psi_t$ with any other distortion function will actually distort  the distribution of the  P\&L tails. We also note that from mathematical point of view $\rho^\psi$ is closely related to Choquet  integrals. The general structure of this class of risk measures is well studied in the static case, originated with the eminent Kusoka result on law-invariant convex risk measures \cite{Kusuoka2001}; see also \cite{Shapiro2013} and \cite[Chapter~4]{FollmerSchiedBook2004}. In particular, it was proved that the class of (static) law-invariant, comonotone coherent risk measures  coincides with the class of coherent risk measures generated by distortion functions. Moreover, this class of risk measures corresponds to the so-called weighted value of risk ($\wvar$) measures; see, for instance, \cite{Cherny2006} for a comprehensive study of these measures in static case and their applications to risk management. We also mention \cite{BernardEtAl2022} for a discussion of distortion risk measures in the context of model uncertainty. Using a similar axiomatic approach, in \cite{ChernyMadan2009} the authors introduce a class of performance measures, called Coherent Acceptability Indices (CAI), as functions  $\alpha:L^\infty\to [0,+\infty]$, that are monotone increasing, scale invariant and quasi-concave. Under some technical continuity assumptions, it can be shown that any CAI $\alpha$ can be characterized by a family of CRMs $\set{\bar\rho^x, \ x\geq 0}$. In particular,  each $\bar\rho^x$ can be generated by a distortion function. Applications of these types of performance measures go beyond the risk management, and were successfully applied, for example,  in portfolio management and pricing of derivatives (cf. \cite{MadanSchoutens2016} and references therein).

Another significant avenue of research is focused  on extending the risk and performance measures to a dynamic setup,  either in  discrete or continuous time framework. Naturally, a dynamic risk or performance measure should take into account the flow of information, which can be achieved by considering the conditional (in probability sense) versions of the corresponding properties such as monotonicity, positive homogeneity, cash-additivity, scale invariance, quasi-concavity, etc. Alternatively, one can `condition' at each time $t$ a given representation of a CRM or CAI, and then study its properties. However, we emphasize  that  a dynamic risk measure  is not just a sequence of conditional static risk measures adapted to the underlying filtration.  Additionally, one must address the issue of measuring riskiness and/or holding preferences consistently over time. This is typically achieved by imposing an additional axiom called time consistency. There is an extensive literature on DCRMs (see survey \cite{AcciaioPenner2010}), and more recently on DCAIs, as well as time consistency in decision making in general. With very few exceptions, for dynamic risk measures, the strong form of time consistency was invoked: if $\rho_{t+1}(X)=\rho_{t+1}(Y)$, then $\rho_t(X)=\rho_t(Y)$.
In addition to admitting a  reasonable economic interpretation, using properties of DCRM, this form of time consistency can be written as a recursive relationship $\rho_t(X) = \rho_t(-\rho_{t+1}(X))$, which can be conveniently used in stochastic control problems with DCRM criteria to derive the Bellman equations, or to link dynamic risk measures to BSDEs or BS$\Delta$Es. On the other hand, there are many other forms  of time consistency, some suitable for some classes of measures or applications, while others in principle can not be satisfied by some risk or performance measures. For example, as shown in \cite{BCZ2010,BCDK2013}, a DCAI can never be strong time consistent. From decision making point of view, time consistency is at the heart of the inter-temporal problems.  We refer the readers to the survey \cite{BCP2017}, where the authors give a comprehensive study of various types of time consistency for both dynamic risk and performance measures, in a discrete time setup, based on the unified approach to time-consistency introduced in \cite{BCP2017a}.

With this at hand, one may ask: what is the dynamic counterpart of $\rho^\psi$, and of DCAI generated by distortion functions? What forms of time consistency do these measures satisfy? This paper is the first attempt to investigate these questions. In Section~\ref{sec:prelim} we start introducing the notations, and give some preliminary results about distortion functions. In Section~\ref{sec:DCRM}, we define and study DCRMs and DCAIs generated by distortion functions. It is tempting, to use representation \eqref{eq:wvar-static}, and define the dynamic $\rho^\psi$ as follows. Let $(\Omega, \sF, \{\sF_t\}_{t \in \mathcal{T}}, \mathbb{P})$ be a filtered probability space, where the increasing collection of $\sigma$-algebras $\sF_t$, $t \in \mathcal{T}$, models the flow of information accumulated through time, and take $X\in L^\infty(\sF_T)$, that models the terminal P\&L of a given investment. Denote by $\bar L^0(\sF_t)$ the set of all $\sF_t$-measurable random variables taking values in $[-\infty,+\infty]$. Define the risk measure $\rho_t^\psi:L^\infty(\sF_T)\to \bar L^0(\sF_t)$, as
 	\begin{equation} \label{eq:ChoquetDCRM-Intro}
	\rho_t^\psi(X)= - \int_{\mathbb{R}} y \dif \psi(\bP(X \leq y\mid\sF_t)).
\end{equation}
We prove that $\rho^\psi_t$ indeed satisfies the axioms of DCRM. Also here we extend the notion of Weighted Value at Risk ($\wvar$) to the dynamic setup, and establish the correspondence between DCRM generated by distortions and the dynamic $\wvar$. Also here, we present the definition and properties of DCAIs generated by families of distortion functions. Section~\ref{sec:families} is dedicated to four important examples of distortion functions widely used in risk management, MINVAR, MAXVAR, MINMAXVAR and MAXMINVAR  \cite{ChernyMadan2009}. For distortion functions corresponding to MINVAR and MAXVAR we derive some new representations. Time consistency - the central topic of this paper - is investigated in Section~\ref{sec:timeconsis}. We first note that thanks to \cite{KupperSchachermayer2009}, for any $\psi$ different from identity, $\rho^\psi$ is not strongly time consistent, and hence majority of existing theoretical results on DCRM do not apply. We prove that $\rho^\psi$ is sub-martingale time consistent, but it is not super-martingale time consistent (see Section~\ref{sec:supermartingale}). Moreover, we show that DCRMs generated by distortion functions are not even weakly acceptance time consistent; Section~\ref{sec:weakacceptance}. In Section~\ref{sec:timeConsistencyDCAI}  we study the time consistency of DCAIs. We collect in the Appendix~\ref{append:defs} some auxiliary definitions and results.

Finally we note that for the sake of brevity, the proofs of some results are deferred  to  the Supplement~\ref{sec:supp}. Although these results are new and important, in the authors opinions their proofs are standard, albeit lengthy and sometimes technical, and hence presented as a technical online supplement.

\section{Preliminaries}\label{sec:prelim}
Let $T$ be a fixed and finite time horizon, and let $\mathcal{T}:= \{0, 1, ..., T \}$. We consider a filtered probability space $(\Omega, \sF, \{\sF_t\}_{t \in \mathcal{T}} , \mathbb{P})$, with $\sF_0 = \{ \emptyset, \Omega \}$ and $\sF = \sF_T$. Throughout, we will use the notations $L^\infty:= L^\infty(\Omega, \sF, \mathbb{P})$, $L^\infty_t:= L^\infty(\Omega, \sF_t, \mathbb{P})$, $t\in\cT$, and  $L^\infty_{t, +}$ the set of all non-negative random variables in $L^\infty_t$, for $t\in\cT$. As usual, all equalities and inequalities will be understood in $\mathbb{P}$-almost surely sense unless otherwise stated.    The set of all probability measures on $[0,1]$ is denoted by $\cP[0,1]$.

We recall that for any real-valued random variable $X : (\Omega, \sF) \rightarrow (\mathbb{R}, \mathcal{B}(\mathbb{R}))$ there exists a \textit{regular conditional distribution} of $X$ given the $\sigma$-algebra $\sF_t$. That is,  there exits a null set $N^X_t \in \sF_t$, such that for any $ \omega \in \Omega \backslash N^X_t$ and any $B\in \mathcal{B}(\mathbb{R})$, we have that $\mathbb{P}(X \in B | \sF_t)(\omega)$ is a distribution function (cf. \cite[Theorem 8.29]{Klenke2013}). In what follows, conditional probabilities will be understood in this sense, and since the set $N^X:= {\cup}_{t \in \mathcal{T}} N^X_t$ is also a null set, and we will use it conveniently instead of $N_t$, for all $t\in\cT$.

\begin{definition}
	A non-decreasing mapping $\psi:[0,1]\rightarrow [0,1]$ is called a \textit{distortion function} if   $\psi(0)= 0$,  $\psi(1)= 1$.
	A distortion function is \textit{regular} if it is concave and continuous. The set of regular distortion functions will be denoted by $\Upsilon$.
\end{definition}

Next result gives a connection between concave and continuous distortion functions, further clarifying the notion of regular distortion function.

\begin{lemma} \label{le:disprop1}
	For any concave distortion function $\psi$ except the identity function, we have that $\psi(x)> x$ for any $x \in (0,1)$. Moreover, 	any concave distortion $\psi$ is continuous on $(0, 1]$.
\end{lemma}
The proof of Lemma~\ref{le:disprop1} is deferred to the Supplement~\ref{sec:supp}. 

\medskip
In view of this result, a concave distortion function which is also continuous at $x=0$ is regular. Unless otherwise stated,
in this work we consider only regular distortion functions.

We will be also using a particular type of distortion functions given by
\begin{equation}\label{eq:regdistortion}
	\psi_{\mu}(y) := \int_{[0,y]} \int_{(z,1]} \frac{1}{s} \mu(\dif s) \dif z,  \quad y \in [0, 1],
\end{equation}
for some $\mu\in\cP[0,1]$. We refer the reader to \cite{Cherny2006,FollmerSchiedBook2004} for properties of this class of distortions.

\begin{remark}\label{rem:psi2mu}
It can be shown (cf. \cite[Lemma 4.63]{FollmerSchiedBook2004}) that \eqref{eq:regdistortion} establishes a one-to-one correspondence between probability measures $\mu\in\cP[0,1]$ and concave distortion functions $\psi$. Moreover,  $\psi_{\mu}$ is continuous on $[0, 1]$ if and only if  $\mu(\{0 \})= 0$. Given a concave distortion function $\psi$, the probability measure $\mu$ such that $\psi_{\mu}= \psi$ has the induced distribution function $F$ on $\left([0, 1], \mathcal{B}([0, 1])\right)$ given by
\begin{align}\label{eq:distPsi}
	\begin{split}
		F(y) =
		\begin{cases}
			0, & y=0\\
			\psi(y)- y \left(\psi_{+}\right)^{\prime}(y), & 0<y<1\\
			1, & y=1,
		\end{cases}
	\end{split}
\end{align}
where $ \psi_{+}^{\prime}$ is the right derivative. In particular, if $\left(\psi_{+}\right)^{\prime}(1-)> 0$, $F$  is discontinuous at point $1$ with jump size $\left(\psi_{+}\right)^{\prime}(1-)$.
\end{remark}

\section{DCRM and DCAI generated by distortion functions}\label{sec:DCRM}
In this section we will introduce classes of dynamic risk measures and dynamic performance measures generated by distortion functions, and investigate various convenient representations of such measures.

In the context of risk and performance measures, the elements $X\in L^\infty$ should be viewed as terminal discounted Profit and Losses (P\&L) of a financial position or portfolio.
A risk measure is meant to determine the riskiness of such positions, measured in the same currency units as $X$.

\begin{definition}\label{def:DCRM}
	A \textit{Dynamic Coherent Risk Measure} (DCRM) is a function $\rho: \mathcal{T} \times L^\infty \times \Omega \rightarrow \mathbb{R}$ that is adapted, normalized, local, cash-additive, monotone decreasing, sub-additive and positive homogeneous (see Appendix~\ref{append:defs} for details).
\end{definition}

Usually, an additional property called time consistency is imposed on DCRMs. The importance of time consistency property is imperative - it relates the agent's preferences through time in a consistent way. One of the main goals of this paper is to investigate this property for a large class of DCRMs and DCAIs; see  Section~\ref{sec:timeconsis}. We refer the reader to \cite{AcciaioPenner2010} for a servey of time consistency of DCRMs, and more recently to  the survey \cite{BCP2017}, where the authors give a comprehensive study of various types of time consistency for both dynamic risk measures and performance measures, in a discrete time setup.

Since the pioneering works \cite{ArtznerDelbaenEberHeath1997,ArtznerDelbaenEberHeath1999}, considerable body of literature was devoted to describe DCRMs and other similar class of functions satisfying certain set of desired axioms under names of convex risk measures, monetary risk measures, acceptability indices, as well as their dynamic counterparts; cf. \cite{DrapeauKupper2010,BCDK2013} for some general results an literature review. Standard results in this field would be dual representations sometimes also called robust representations. For example, it can be proved that in a static setup, $\cT=\set{0,1}$, a coherent risk measure admits the following representation
\begin{align} \label{eq:defstaticRM}
	\rho(X)=-\inf_{\mathbb{Q} \in \mathcal{Q}} \mathbb{E}_\mathbb{Q}(X),
\end{align}
where $\mathcal{Q}$ is a nonempty set of probability measures, absolutely continuous with respect to $\mathbb{P}$. Different sets $\cQ$ would generated different risk measures. The robust representation \eqref{eq:defstaticRM}, in addition  to  being useful for theoretical developments, also  gives alternative interpretations of risk measures. For example, the set of probabilities $\cQ$ can be viewed as plausible scenarios, in which case, the value of $\rho(X)$ is equal to the worst expected P\&L under possible scenarios. On the other hand, using \eqref{eq:defstaticRM} for efficient computations could be challenging, unless the extreme measure $\bQ$ is known. In many applications, one would usually use risk measures given by explicit analytical formulas. One large and important class of such risk measures are DCRMs generated by distortion functions, given by
\begin{equation}\label{eq:wvar-static2}
	\rho^\psi(X) = - \int_{\mathbb{R}} y \dif \psi(\mathbb{P}(X \leq y)).
\end{equation}
It is clear that $\rho^\psi_t$ with $\psi(x)=x$ becomes the negative of (linear) expectation, and for any other $\psi\in\Upsilon$ the  distribution of the tails of the random variable $X$ will be distorted, which in the context of risk management corresponds to distorting  distribution of  the large or extreme losses and gains. The general structure of this class of risk measures \textit{in the static setup} is well studied, originated with the eminent Kusoka result on law-invariant convex risk measures \cite{Kusuoka2001}; see also \cite{Shapiro2013} and \cite[Chapter~4]{FollmerSchiedBook2004}. In particular, it was proved that the class of (static) law-invariant, comonotone coherent risk measures  coincides with the class of coherent risk measures generated by distortion functions. Moreover, it is known that this class of risk measures corresponds to the so-called weighted value of risk ($\wvar$) measures; see, for instance, \cite{Cherny2006} for a comprehensive study of these measures in static case and their applications to risk management. However, to the best of our knowledge, the dynamic counterpart of \eqref{eq:wvar-static2} is not formally studied, and first goal of this paper is to fill this gap.

Next, we introduce the key object in this paper, that, as seen below, corresponds to the dynamic version of \eqref{eq:ChoquetDCRM-Intro}. For any $\psi\in\Upsilon,  X \in L^\infty$  and $t \in \mathcal{T}$, we define
\begin{align} \label{eq:ChoquetDCRM}
  \rho_t^\psi(X):= \int_{[0,\infty)}   \psi \left(\bP(-X > y \mid \sF_t)\right)   \dif y+ \int_{(-\infty,0)}  \left[ \psi(\bP(-X > y \mid \sF_t)) -1 \right]  \dif y,
\end{align}
which can be viewed as the conditional Choquet integral \cite[Section~4.6]{FollmerSchiedBook2004}.

\begin{proposition} \label{prop:DCRM}
For any $\psi\in\Upsilon$, the mapping $\rho^\psi$ is a law-invariant DCRM.
\end{proposition}
\begin{proof} 	
Law-invariance of $\rho_t^\psi$ is clear.
Verification of properties of DCRM follow by somewhat standard but lengthy arguments. The main technical difficulties are related to the measurability issues and giving sense to conditional quantities through regular conditional distributions.  The detailed proof is deferred to the Supplement~\ref{sec:supp}
\end{proof}

We call $\rho^\psi$ the \textit{DCRM generated by the distortion function} $\psi$. Using integration by parts,  DCRMs generated by regular distortion functions admit the following alternative representation
 	\begin{equation} \label{eq:ChoquetDCRM-2}
 		 	\rho_t^\psi(X)= - \int_{\mathbb{R}} y \dif \psi(\bP(X \leq y\mid\sF_t)).
 	\end{equation}

Here we briefly extend some of the relevant results to the dynamic setup.

\begin{definition} \label{def:2-21}
   Given a probability measure $\mu\in\cP[0,1]$, the \textit{dynamic weighted value at risk} (dWV@R) is defined as
 \begin{align}\label{eq:dwvar-def}
   	\dwvar^\mu_t(X):= \int_{(0,1]}  \avar_{\alpha}(X\mid\sF_t) \mu(\dif \alpha),   \qquad X \in L^\infty, \ t \in \mathcal{T},   	
\end{align}
where $\avar_{\alpha}$ is the conditional average value at risk.
\end{definition}
The conditional average value at risk is defined by analogy to the regular average value at risk through the conditional quantiles or value at risk.
The delicate part is to give proper probabilistic meaning to these quantities, which we briefly present in Appendix~\ref{append:defs}.  Detailed proofs of some relevant properties are deferred to the Supplement~\ref{sec:supp}. 

Next result gives alternative representation of $\dwvar^\mu(X)$ by means of conditional quantile function. Please refer to Definition A.1 in the Appendix 1 for the definition of $ q^{+}_{z}(X \mid \sF_t)$.

\begin{lemma} \label{le:altdwar}
Let  $\mu\in\cP([0,1])$ and denote by $ \psi_{\mu, +}^{'}$ the right hand derivative of $\psi_{\mu}$. Then,
\begin{align} \label{eq:2-54}
	   	\dwvar^\mu_t(X)(\omega) &= -\int_{(0,1)}   q^{+}_{z}(X \mid \sF_t)(\omega) \psi_{\mu, +}^{'}(z)  \dif z,
\end{align}
for any $X \in L^\infty$, $t \in \mathcal{T}$ and $ \omega \in \Omega \backslash N^X $.
\end{lemma}
\begin{proof} In view of \eqref{eq:avar} we have
\begin{align*}
	\dwvar^\mu_t(X)(\omega) 
  &=  -\int_{(0,1]} \frac{1}{\alpha}  \int_{(0,\alpha)}  q^{+}_{z}(X \mid \sF_t)(\omega)  \dif z \mu(\dif \alpha)\\
    &=  -\int_{(0,1)}  q^{+}_{z}(X \mid \sF_t)(\omega) \int_{(0,1]}  \frac{1}{\alpha}  \1_{z< \alpha}  \mu(\dif \alpha) \dif z\\
    &=  -\int_{(0,1)}  q^{+}_{z}(X \mid \sF_t)(\omega)   \psi_{\mu, +}^{'}(z)  \dif z.
\end{align*}
\end{proof}

Next, we prove that  $\dwvar^\mu$ is a DCRM generated by a distortion function.
\begin{theorem}	 \label{thm:dwar}
Let  $\mu$ be a given probability measure on $[0, 1]$. Then, for any $X \in L^\infty$, $t \in \mathcal{T}$,  and $ \omega \in \Omega \backslash N^X $,
   \begin{equation} \label{eq:disdwar}
   	\dwvar^\mu_t(X)(\omega)= \rho^{\psi_{\mu}}_t(X)(\omega).
   \end{equation}
Hence,   $\dwvar^\mu(X)$ is a law-invariant DCRM.
\end{theorem}
\begin{proof}
 First we consider $X \geq 0$. By Lemma~\ref{le:altdwar} and the  definition of conditional upper quantile, we have
\begin{align*}
	\dwvar^\mu_t(X)(\omega)  
    &= -\int_{(0,1)} \sup \{y \in \mathbb{R} \mid \bP(X < y\mid\sF_t)(\omega) \leq z \} \psi_{\mu, +}^{'}(z)   \dif z\\
    &= -\int_{(0,1)}   \int_{[0,\infty)}  \1_{\bP(X < y\mid\sF_t)(\omega) \leq z } dy   \psi_{\mu, +}^{'}(z)   \dif z \\
    &= -\int_{[0,\infty)}  \int_{(0,1)} \1_{\bP(X < y\mid\sF_t)(\omega) \leq z}  \psi_{\mu, +}^{'}(z)   \dif z \dif y  \\
    &= -\int_{[0,\infty)} \left[ 1- \psi \left( \bP(X < y\mid\sF_t)(\omega) \right) \right] \dif y  \\
    &= \int_{(-\infty, 0)} \left[ \psi \left( \bP(-X > y\mid\sF_t)(\omega) \right)- 1 \right] \dif y  \\
    &= \rho_t^{\psi_{\mu}}(X)(\omega).
\end{align*}
If $X \in L^\infty$, we consider $X^{\prime}=X+C \geq 0$, where $C= \esssup X $. Hence, from the above, combined with cash-additivity of $\rho_t^{\psi_{\mu}}$, yield
$$
\dwvar^\mu_t(X') = \rho_t^{\psi_{\mu}}(X) - C.
$$
On the other hand, by \eqref{eq:dwvar-def} and cash-additivity of $\avar_{\alpha}$, we get
$$
\dwvar^\mu_t(X') =  \dwvar^\mu_t(X)  - C.
$$
This concludes the proof.
\end{proof}

\begin{corollary}
Theorem~\ref{thm:dwar}  implies  that any DCRM generated by a distortion function is also a $\dwvar^\mu$ corresponding to a probability measure $\mu$. Hence, all financial interpretations of  $\dwvar^\mu$ apply to $\rho^{\psi_\mu}$, and additionally, \eqref{eq:dwvar-def} and \eqref{eq:2-54} provide alternative methods to compute $\rho^\psi$.
\end{corollary}

The notion of static Coherent Acceptability Index (CAI) was introduced in \cite{chernyMadan2006}, in the same axiomatic spirit of CRM, aiming to generalize the known performance measures such as Sharpe ratio, gain to loss ratio, risk adjusted return on capital, etc. The dynamic coherent acceptability indices (DCAIs), in a discrete time setup, were first introduced and studied in \cite{BCZ2010}, and consequently in \cite{BCIR2012,BCC2014,BCDK2013,BiaginiBion-Nadal2012}.

\begin{definition} A DCAI is a function $\alpha: \mathcal{T} \times L^\infty \times \Omega \rightarrow [0, \infty]$ that is adapted, local, quasi-concave, monotone increasing and scale invariant.
\end{definition}

Similar to DCRM, we study time consistency of DCAI separately, in next section. Under some technical continuity assumptions, it was proved that a DCAI $\alpha_t$ can be uniquely identified with a family $\rho^x_t, \ x\in[0,+\infty)$, of DCRMs. We focus our attention on the case when $\rho^x, x\geq0$, are generated by distortion functions. This is motivated, in particular, by the extensive use of the static CAIs generated by distortion function in various applications from finance (cf. \cite{MadanSchoutens2016}).

Let us introduce some definitions and discuss some properties of families of distortion functions and the corresponding DCRMs.
In what follows, we assume that the underlying probability space is atomless.

\begin{definition}
	A family, indexed by $x >0$, of distortion functions $(\psi_x)_{x>0}$, is called increasing, if $\psi_{x_1}(y) \leq \psi_{x_2}(y)$ for any $y \in [0, 1]$ and  $0< x_1 \leq x_2$.
	Respectively, it is called right continuous, if $\lim\limits_{{z \rightarrow x^{+}}} \psi_{z}(y) = \psi_x(y)$ for all $y \in [0, 1]$ and $x>0$. 
\end{definition}

\begin{definition}
	A family, indexed by $x >0$, of DCRMs $(\rho^{\psi_x})_{x >0}$, is called increasing, if $\rho^{\psi_{x_1}}_t(X) \leq \rho^{\psi_{x_2}}_t(X)$ for any $X \in L^\infty$, $t \in \mathcal{T}$ and  $0< x_1 \leq x_2$.
Respectively, 	it is called  right continuous, if $\lim\limits_{{z \rightarrow x^{+}}} \rho^{\psi_z}_t(X) = \rho^{\psi_x}_t(X)$ for all $X \in L^\infty$, $t \in \mathcal{T}$ and $x>0$.
\end{definition}

\begin{lemma}   \label{le:cexin}
 A family of DCRMs $(\rho^{\psi_x})_{x >0}$ is increasing if and only if $(\psi_{x})_{x >0}$ is an increasing family of regular distortion functions.
 Moreover, 	a family of DCRMs $(\rho^{\psi_x})_{x >0}$ is right continuous if and only if  $(\psi_{x})_{x >0}$ is a right continuous family of regular  distortion functions.
\end{lemma}
The proof of Lemma~\ref{le:cexin} can be found in accompanying supplement. 

\begin{remark} Several notes are in order:
	\begin{itemize}
\item The assumption  made above  that probability space is atomless means that a probability space supports random variables with continuous distribution if and only if such probability space is atomless; for details see \cite[Proposition A.27]{FollmerSchiedBook2004}. We use this in the proof of Lemma~\ref{le:cexin} to ensure that $X= 1$ with probability $y^{\star}$ exists. If atomless property does not hold, then generally speaking the `only if' directions in Lemma~\ref{le:cexin} does not hold either. However, Lemma~\ref{le:cexin} may still hold true for a probability space with atoms, if the atoms do not affect the distribution of $X$.

\item A family of DCRMs $(\rho^{\psi_x})_{x >0}$ being increasing does not necessarily imply that the family of $(\psi_{x})_{x >0}$ is increasing; see Example~B.10. 
\item A family of DCRMs $(\rho^{\psi_x})_{x >0}$ being  right continuous does not necessarily imply that $(\psi_{x})_{x >0}$ is a right continuous family of distortion functions; see Example~B.10.
	\end{itemize}
\end{remark}

Given an increasing family of distortion functions $\Psi= (\psi_x)_{x> 0}$.  For any $t \in \mathcal{T}$,  $X \in L^\infty$ and $\omega \in \Omega \backslash N^X$, we define
\begin{equation}  \label{def:DCAI}
\alpha_{t}^{\Psi}(X)(\omega):= \sup \left\{x \in \mathbb{R}_+ \mid \rho_t^{\psi_{x}}(X)(\omega) \leq 0 \right \},
\end{equation}
with convention $\sup \emptyset= 0$. Moreover, for any $\omega \in N^X$ we let $\alpha_{t}^{\Psi}(X)(\omega)= 0$.

\begin{proposition}
 The mapping  $\alpha_{t}^{\Psi}, t\in\cT$, is a law-invariant DCAI.
\end{proposition}
The proof of this result, thanks to the robust representations results for DCAIs, follows from 	\cite{BiaginiBion-Nadal2012} or \cite{BCDK2013}. For the sake of completeness, a direct proof by verifying the required properties of DCAI is also presented in the Supplement~\ref{sec:supp}.

\subsection{Families of DCRMs} \label{sec:families}

In this section we give several important examples of families of DCRMs that generate DCAIs. These examples are based on families of distortion functions introduced in \cite{ChernyMadan2009}, that lead to the so-called MINVAR, MAXVAR, MAXMINVAR and MINMAXVAR risk measures. As already mentioned, the corresponding CAIs were successfully used, in a static setup, in a series of papers and monographs devoted to conic finance \cite{MadanSchoutens2016,MadanSchoutensEq2011,MadanSchoutens2011a}. Here we will present the dynamic counterpart of these measures.

\subsubsection{Dynamic MINVAR}
Consider the following family of regular distortion functions
\[
	\psi_{x}^{\mathrm{MINVAR}}(y)=1-(1-y)^{x+1}, \quad x \in \mathbb{R}_{+}, \ y \in[0,1].
\]
Clearly this is an increasing and continuous family of distortions. As next result shows, the probability measure $\mu$ from \eqref{eq:regdistortion} can be computed explicitly for these examples.
\begin{proposition}
	Given $x \in \mathbb{R}_{+}$, let $\mu\in\cP[0,1]$ be such that $\psi_{\mu}= \psi_{x}^{\mathrm{MINVAR}}$. Then, $\mu$ follows $\mathrm{Beta}(2, x)$ distribution.
\end{proposition}
\begin{proof}
Let $F$ denote the distribution function of $\mu$. By \eqref{eq:distPsi}, for any $y\in [0, 1]$, we have
	\begin{align*}
		F(y) &=1-(1-y)^{x+1} - (x+1)y (1-y)^x= 1-(1-y)^x -xy (1-y)^x \\
		&= 1-(1-y)^x -xy(1-y)^x = I_y(1,x)-xy(1-y)^x\\
		&= I_y(1,x)-\frac{y(1-y)^x}{B(1,x)}= I_y(2,x),
	\end{align*}
	where $I_y(2,x)$ is the regularized incomplete beta function and $B(1,x)= \frac{\Gamma(1) \Gamma(x)}{\Gamma(1+x)}= \frac{1}{x}$.
Note that $F$ is exactly the cumulative distribution function of the beta distribution $\mathrm{Beta}(2,x)$.
\end{proof}

Similar to \cite{ChernyMadan2009}, we provide an intuitive interpretation of $\rho_t^{\psi_x^\mathrm{MINVAR}}$.
Namely,  for any $x \in \mathbb{Z}_+$ and $t \in \mathcal{T}$, it holds that
\begin{align*}
	\rho^{\psi_{x}^{\mathrm{MINVAR}}}_t(X)=-\bE[Y\mid \sF_t], \quad \textrm{ with } \ Y \stackrel{d}{=} \min \left\{X_{1}, \ldots, X_{x+1}\right\},
\end{align*}
where  $X_{1}, \ldots, X_{x+1}$  are identically distributed as $X$ and mutually conditionally independent given $\sF_t$.
Indeed,
\begin{align*} 
	\begin{split}
		- \bE[Y\mid \sF_t] &= - \int_{-\infty}^{+\infty} y d \bP(Y \leq y\mid\sF_t)
		= -\int_{-\infty}^{+\infty} y d (1- \bP(\min \left\{X_{1}, \ldots, X_{x+1}\right\} > y\mid\sF_t))\\
		&= -\int_{-\infty}^{+\infty} y d (1- \prod_{i=1}^{x+1} \bP(X_i>y\mid\sF_t))
		= - \int_{-\infty}^{+\infty} y d (1-(1- F_{X\mid\sF_t}(y))^{x+1})\\
		&= -\int_{-\infty}^{+\infty} y d\psi_{x}^{MINVAR}\left(F_{X\mid\sF_t}(y)\right)
		= \rho^{\psi_{x}^{MINVAR}}_t(X).
	\end{split}
\end{align*}
The DAI generated by $\rho^{\psi_{x}^{\mathrm{MINVAR}}}$ is denoted by $\alpha^{\mathrm{MINVAR}}$.

\subsubsection{Dynamic MAXVAR}
The DCRM called dynamic MAXVAR is generated by the following family of  distortions
\[
	\psi_{x}^{MAXVAR}(y)=y^{\frac{1}{x+1}}, \quad x \in \mathbb{R}_{+}, \ y \in[0,1],
\]
with the corresponding $\mu\in\cP[0,1]$ from  \eqref{eq:regdistortion} given in the next result.
\begin{proposition}
For a fixed $x \in \mathbb{R}_{+}$, let $\mu$ be the probability measure such that $\psi_{\mu}= \psi_{x}^{\mathrm{MAXVAR}}$. Then,
$$
\mu= \frac{x}{x+1} \hat{\mu} + \frac{1}{x+1} \delta_1,
$$
where $\hat{\mu}$ follows the power with density $f_{\widehat \mu}(y) \propto y^{-\frac{x}{x+1}}$, and $\delta_1$ is the Dirac measure at $1$.
\end{proposition}
\begin{proof}
By \eqref{eq:distPsi}, the distribution function of $\mu$ is given by
\[
F(y) =
	\begin{cases}
		0, & y=0\\
		\frac{x}{x+1}y^{\frac{1}{x+1}}, & 0<y<1\\
		1, & y=1.
	\end{cases}
\]
	Note that $F$ is discontinuous at $1$ since,
	\begin{align*}
		\mu(\{ 1 \})= \left( \psi_{x}^{\mathrm{MAXVAR}} \right)^{\prime}(1-) = \frac{1}{x+1} y^{\frac{1}{x+1}- 1}\vert_{y= 1-}=  \frac{1}{x+1}.
	\end{align*}
This completes the proof. 
\end{proof}

Next, we will show that for fixed $x \in \mathbb{Z}_+$ and $t \in \mathcal{T}$,
\[
	\rho^{\psi_{x}^{\mathrm{MAXVAR}}}_t(X)=-\bE[Y\mid \sF_t], \quad \textrm{with}\  X \stackrel{d}{=} \max \left\{Y_{1}, \ldots, Y_{x+1}\right\},
\]
where  $Y_{1}, \ldots, Y_{x+1}$  are identically distributed as $Y$ and mutually conditionally independent given $\sF_t$. This, in particular, justifies the name $\mathrm{MAXVAR}$; see also \cite{ChernyMadan2009}. Indeed, since
\[ 
\begin{split}
F_{X|\sF_t}(y)&= \bP(\max \left\{Y_{1}, \ldots, Y_{x+1}\right\} \leq y\mid\sF_t)
		= \bP(Y_{1} \leq y, \ldots, Y_{x+1}\leq y \mid \sF_t)\\
		&= \prod_{i=1}^{x+1} \bP(Y_i \leq y\mid\sF_t)
		= (F_{Y|\sF_t}(y))^{x+1},
	\end{split}
\]
we have that $F_{Y|\sF_t}(y)= (F_{X|\sF_t}(y))^{\frac{1}{x+1}}= \psi_{x}^{\mathrm{MAXVAR}}(F_{X|\sF_t}(y))$. Hence,
\begin{align*}
	- \bE[Y\mid \sF_t] &= - \int_{-\infty}^{+\infty} y d \bP(Y \leq y\mid\sF_t)
	= -\int_{-\infty}^{+\infty} y d\psi_{x}^{\mathrm{MAXVAR}}\left(F_{X|\sF_t}(y)\right)\\
	&= \rho^{\psi_{x}^{\mathrm{MAXVAR}}}_t(X).
\end{align*}

\subsubsection{Dynamic MAXMINVAR and MINMAXVAR}
Similarly to the previous examples, one can elevate the MAXMINVAR and MINMAXVAR defined in \cite{ChernyMadan2009} to the dynamic (conditional) setup, respectively considering the families of distortion functions
\begin{align*}
	\psi_{x}^{\mathrm{MAXMIN}}(y)=\left(1-(1-y)^{x+1}\right)^{\frac{1}{x+1}}, \quad
	\psi_{x}^{\mathrm{MINMAX}}(y)=1-\left(1-y^{\frac{1}{x+1}}\right)^{x+1},
\end{align*}
where $x \in \mathbb{R}_{+}, \ y \in[0,1]$.

Moreover, one can prove that for a fixed $x \in \mathbb{Z}_+$ and $t \in \mathcal{T}$,
\begin{align} 
	\rho^{\psi_{x}^{\mathrm{MAXMIN}}}_t(X)=-\bE[Y\mid \sF_t], \quad \textrm{with}\  \max \left\{Y_{1}, \ldots, Y_{x+1}\right\} \stackrel{d}{=} \min \left\{X_{1}, \ldots, X_{x+1}\right\},
\end{align}
where  $X_{1}, \ldots, X_{x+1}$  are identically distributed as $X$ and $Y_{1}, \ldots, Y_{x+1}$ are identically distributed as $Y$,
and both sequences are mutually conditionally independent given $\sF_t$, and respectively,
$$
	\rho^{\psi_{x}^{\mathrm{MINMAX}}}_t(X)=-\bE[Y\mid \sF_t], \quad \textrm{with}\  Y \stackrel{d}{=} \min \left\{Z_{1}, \ldots, Z_{x+1}\right\}, \  X \stackrel{d}{=}  \max \left\{Z_{1}, \ldots, Z_{x+1}\right\},
$$
where $Z_{1}, \ldots, Z_{x+1}$  are random variables identically distributed as $Z$ and mutually conditionally independent given $\sF_t$.

Finally, we want to mention that $\mu\in\cP[0,1]$ corresponding to $\psi_{x}^{\mathrm{MAXMIN}}$ or $\psi_{x}^{\mathrm{MINMAX}}$ does not follow a classical distribution, hence they are not presented here.

\section{Time consistency of measures generated by distortion functions} \label{sec:timeconsis}

Time consistency in decision making is a large topic with many existing fundamental results. We refer to the survey  \cite{BCP2017}, where the authors present a comprehensive literature review and discuss different frameworks to study time consistency for risk or performance measures that admit numerical representations. Here, we will use the approach to time consistency developed in \cite{BCP2017a}. For convenience, we list the relevant forms of time consistency in Appendix~\ref{append:defs}. We assume the setup of Section~\ref{sec:DCRM}, and in particular we consider an atomless filtered probability space.

\subsection{Time consistency of $\rho^\psi$}

We recall that if $\psi= \textrm{Id}$  is the identity function, then $\rho_{t}^\psi(X)= \mathbb{E}(-X \mid \sF_t)$. This extreme case naturally satisfies all forms of time consistency. Therefore, in this section we omit this trivial case, and unless otherwise specified,  we assume that $\psi\neq \textrm{Id}$.

As already mentioned, $\rho^\psi$ is never strong time consistent (i.e. recurrent), except for $\psi$ identity, but as we will show, it satisfies weaker forms of time consistency.

\subsubsection{Sub-martingale and weakly rejection time consistency}\label{sec:weak-rejection-rho}

\begin{theorem} \label{th:submtg}
For any $\psi \in \Upsilon$, the DCRM $\rho^\psi$ is sub-martingale time consistent.
\end{theorem}
\begin{proof}
Let $t,s\in\cT$ such that $t \leq s$. Then, first by  Fubini's theorem and then by Jensen's inequality, we deduce
\begin{align*}
\bE(\rho_{s}^\psi(X)\mid\sF_t)
&=  \int_{[0, \infty)}  \bE (\psi(\bP(- X > y|\sF_{s})) |\sF_t) \dif y + \int_{(-\infty, 0)}  \bE( [ \psi(\bP(- X > y|\sF_{s})) -1 ]|\sF_t)  \dif y \\
&\leq  \int_{[0, \infty)}  \psi (\bP(- X > y\mid\sF_{t})) \dif y + \int_{(-\infty, 0)}  [\psi(\bP(- X > y\mid\sF_{t})) -1]  \dif y  \\
&= \rho_{t}^\psi(X).
\end{align*}
The proof is complete.
\end{proof}

In view of \cite[Proposition 5.4]{BCP2017}, weak time consistency is one of the weakest forms of time consistency, in the sense that it is implied by any time consistency generated by a projective update rule, and in particular weak rejection time consistency is implied by the sub-martingale time consistency. Thus, Proposition~\ref{th:submtg} implies that $\rho^\psi$ is weakly rejection time consistent.

\subsubsection{Middle rejection time consistency}

Inspired by \cite[Example 2.7]{MaEtAl2018}, we give a counterexample of  middle rejection time consistency of $\rho^\psi$.

\begin{example} \label{ex:nonmiddle}
Let $\mathcal{T}:= \{0,1,2\}$, assume that $X= (X_t)_{t \in \mathcal{T}}$ is a two-period binomial model shown in the following graph, with the same upward probability $0.5$ at each time step, and let  $\mathbb{F}= (\sF^X_t)_{t \in \mathcal{T}}$ be the natural filtration generated by $X$. Take $\psi(z):= \sqrt{z}$, $z\in[0,1]$,  and let $Y= -\rho_{1}^\psi(X_2) \in L^\infty_1$.

\begin{tikzpicture}[>=stealth,sloped]
    \matrix (tree) [%
      matrix of nodes,
      minimum size=1cm,
      column sep=1cm,
      row sep=0.2cm,
    ]
    {
       $X_0$&$X_1$&$X_2$&$\mathbb{P}$\\
          &   & 2 & 0.25  \\
          & 1 &   & \\
     \  0 &   & 0 & 0.5 \\
          &-1 &   & \\
          &   &-2 & 0.25  \\
    };
    \draw[->] (tree-4-1) -- (tree-3-2);
    \draw[->] (tree-3-2) -- (tree-2-3);
    \draw[->] (tree-4-1) -- (tree-5-2);
    \draw[->] (tree-5-2) -- (tree-6-3);
    \draw[->] (tree-3-2) -- (tree-4-3);
    \draw[->] (tree-5-2) -- (tree-4-3);
\end{tikzpicture}

\noindent Note that on $\{X_1= 1\}$, $\rho^\psi_{1}(X_2)= \sqrt{2}- 2$, while on $\{X_1= -1\}$, $\rho^\psi_{1}(X)= \sqrt{2}$. Thus, $\rho^\psi_{0}(X_2)= \sqrt{3}-1$. On the other hand, by direct evaluations  $\rho^\psi_{0}(Y) = 2\sqrt{2}- 2$.  Thus, $ \rho^\psi_{0}(X_2)< \rho^\psi_{0}(Y) $, which implies that $\rho^\psi$ is not middle rejection time consistent.
\end{example}

\begin{remark}
 Example~\ref{ex:nonmiddle} can also serve as counterexample that $\rho^\psi$ is not strongly time consistent.
\end{remark}

\subsubsection{Super-martingale time consistency}\label{sec:supermartingale}

\begin{proposition} \label{prop:nonsupermtg}
For any non-constant random variable $X \in L^\infty$ and  $\psi\in\Upsilon$, except the identity distortion function,
\[
\rho_{t}^{\psi}(X) >  \mathbb{E}\left[\rho_{T}^{\psi}(X)\mid \sF_{t}\right], \quad \textrm{for all } t \in \mathcal{T}\backslash \{T\}.
\]
\end{proposition}
\begin{proof}
By normalization and cash additivity of $\rho_{t}^{\psi}$, we have $\rho_{T}^{\psi}(X)= -X$, and hence it is enough to show that,
\[
\rho_{t}^{\psi}(X)-  \mathbb{E}\left[-X\mid \sF_{t}\right]> 0, \quad \forall t \in \mathcal{T}\backslash \{T\}.
\]
Let $a$ and $b$ be the essential infimum and essential supremum of $X$ respectively. Since $X$ is not a constant, then  $a<  b$.
If $X$ is continuous, then there exists $E \subset [-b, -a]$ of positive Lebesgue measure such that
$$
\bP(-X > y\mid\sF_t)(\omega) \in (0,1),
$$
for any $y \in E$ and $\omega \in \Omega \backslash N^X$.  If $X$ is discrete, then without loss of generality, assume $X$ is a binary random variable taking values $a$ and $b$ with strict positive probabilities. Hence, for any $y \in [-b,-a)$, \ $\omega \in \Omega \backslash N^X$, we have that $ \bP(-X > y\mid\sF_t)(\omega) = \mathbb{P}(X= a\mid\sF_t)(\omega) \in (0,1) $, and thus we put $E=[-b,-a)$ in this case. 

In view of Lemma~\ref{le:disprop1}, there exists $E \subset [-b, -a]$ with positive Lebesgue measure, so that for any $y \in E$ and $\omega \in \Omega \backslash N^X$ 
\begin{equation} \label{eq:sup-1}
\psi(\mathbb{P}(-X > y\mid\sF_t)(\omega))- \mathbb{P}(-X > y\mid\sF_t)(\omega) > 0. \\
\end{equation}

By \eqref{eq:sup-1}, for any $ t \in \mathcal{T}\backslash \{T\} $ and $\omega \in \Omega \backslash N^X$,
\begin{align*}
 \rho_{t}^{\psi}(X)(\omega)-  & \mathbb{E}\left[-X\mid \sF_{t}\right](\omega)
 = \int_{[0,\infty)}   \psi(\mathbb{P}(-X > y\mid\sF_t)(\omega))   \dif y \\
 & + \int_{(-\infty,0)}  [ \psi(\mathbb{P}(-X > y\mid\sF_t)(\omega)) -1 ]  \dif y-  \mathbb{E}\left[-X\mid \sF_{t}\right](\omega) \\
 =& \int_{[0,\infty)}   \psi(\mathbb{P}(-X > y\mid\sF_t)(\omega))   \dif y + \int_{(-\infty,0)} [ \psi(\mathbb{P}(-X > y\mid\sF_t)(\omega)) -1 ]  \dif y \\
 &-  \int_{[0,\infty)}   \mathbb{P}(-X > y\mid\sF_t)(\omega)  \dif y - \int_{(-\infty,0)}  [ \mathbb{P}(-X > y\mid\sF_t)(\omega) -1 ]  \dif y  \\
 =&  \int_{[0,\infty)}  \left[ \psi(\mathbb{P}(-X > y\mid\sF_t)(\omega))- \mathbb{P}(-X > y\mid\sF_t)(\omega) \right]  \dif y  \\
 & + \int_{(-\infty,0)} \left[  \psi(\mathbb{P}(-X > y\mid\sF_t)(\omega)) -\mathbb{P}(-X > y\mid\sF_t)(\omega)  \right] \dif y \\
 \geq&  \int_{E}  \left[  \psi(\mathbb{P}(-X > y\mid\sF_t)(\omega))- \mathbb{P}(-X > y\mid\sF_t)(\omega) \right]  \dif y   >   0.
\end{align*}
This concludes the proof.
\end{proof}

As a direct consequence of Proposition~\ref{prop:nonsupermtg}, we have the following result on super-martingale time consistency of $\rho^\psi$.
\begin{theorem}
For any  $\psi$ except the identity distortion function, $\rho^\psi$ does not satisfy super-martingale time consistency on $L^\infty$. In other words, for any  $\psi$ except the identity distortion function, there always exists $A \in \sF_t$ with $\bP(A)> 0$, such that on $A$,
\begin{align*}
\rho_{t}^{\psi}(X) >  \bE\left[\rho_{s}^{\psi}(X)\mid \sF_{t}\right],
\end{align*}
for some  $X \in L^\infty$  and some  $t, s \in \mathcal{T}$, $s> t $.
\end{theorem}

\subsubsection{Weak acceptance time consistency}\label{sec:weakacceptance}

In this section we show that, in general, DCRMs generated by distortion functions are not even weakly acceptance time consistent, in the sense of the following result.

\begin{theorem} \label{thm:nonweak}
	For any probability measure $\mu\in\cP[0, 1]\setminus \set{\delta_1}$ such that $\psi_{\mu}$ is absolutely continuous,  there exists a filtered probability space,  such that $\rho^{\psi_{\mu}}$ is not weakly acceptance time consistent.
\end{theorem}
The proof of Theorem~\ref{thm:nonweak} follows from a series of lemmas below.

\medskip
In view of Remark~\ref{rem:psi2mu}, Theorem~\ref{thm:nonweak} is equivalent to the statement that, for any absolutely continuous regular distortion $\psi$ except for the identity function,  there exists a filtered probability space,  such that $\rho^{\psi}$ is  not weakly acceptance time consistent.

\begin{lemma} \label{lemma:nonweak-1}
	For any probability measure $\mu\in\cP([0, 1])\setminus \cP'$, with
\[
		\cP' := \Set{\frac{a-1}{a} \delta_{\frac{1}{a+1}}+ \frac{1}{a} \delta_{1} \Mid  a \geq 1 },
\]	
and such that $\psi_{\mu}$ is absolutely continuous, we have
	\begin{equation} \label{eq:weaktcaux1}
		\psi_{\mu}\left( \int_{(0,1)}  z \psi_{\mu}^{'}(z)  dz \right)+ \int_{(0,1)}  z \psi_{\mu}^{'}(z)  dz-1 < 0.
	\end{equation}
\end{lemma}

\noindent
\textit{Proof.}
Let $m_{\mu}:=  \int_{(0,1)}  z \psi_{\mu}^{'}(z)  dz$. Then, after integration by parts, we obtain

\begin{align}  \label{eq:weaktc-1}
m_{\mu} =  1- \int_{(0,1)}  \psi_{\mu}(z) dz.
\end{align}
With this at hand, \eqref{eq:weaktcaux1}  becomes
\begin{align}  \label{eq:weaktc-2}
		\psi_{\mu}\left( m_{\mu} \right)+ m_{\mu}-1 < 0,
\end{align}
which will be proved in three steps.

We first prove a mild version of \eqref{eq:weaktc-2}, namely
	\begin{equation}\label{eq:weaktc-3}
		\psi_{\mu}\left( m_{\mu} \right)+ m_{\mu}-1 \leq 0.
	\end{equation}
for any $\mu\in\cP([0,1])$ such that $\psi_{\mu}$ is absolutely continuous.

\begin{wrapfigure}{r}{0.5\textwidth}
	\includegraphics[width=0.45\textwidth]{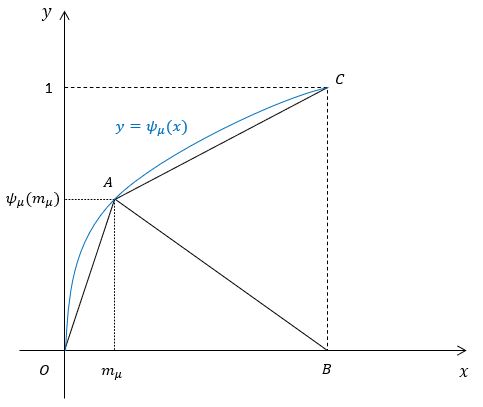}
	\caption{}
	\label{fig:weaktc}
\end{wrapfigure}
	
In view of Figure~\ref{fig:weaktc},  we note that the sum of area of the triangles $\Delta AOB$ and $\Delta ABC$ is $\frac{1}{2} \psi_{\mu}\left( m_{\mu} \right)+ \frac{1}{2}(1-m_{\mu})$. On the other hand, due to the concavity of function $\psi_{\mu}$, this area is smaller than the area under the curve  $y=\psi_{\mu}(x)$, for $x\in[0,B]$, and thus
	$$
		\int_{(0,1)}  \psi_{\mu}(z) dz \geq \frac{1}{2} \psi_{\mu}\left( m_{\mu} \right)+ \frac{1}{2}(1-m_{\mu}).
	$$
By \eqref{eq:weaktc-1},  the last inequality implies that
$$
1- m_{\mu} \geq \frac{1}{2} \psi_{\mu}\left( m_{\mu} \right)+ \frac{1}{2}(1-m_{\mu}),
$$
and thus \eqref{eq:weaktc-3} is proved.

In the next step we will show that equality in \eqref{eq:weaktc-3} holds if and only if $\mu\in\cP'$.
	
Assume that equality in \eqref{eq:weaktc-3} holds. Clearly, in view of the above and Figure~\ref{fig:weaktc}, this is equivalent that the graph of $y=\psi_{\mu}(x)$, coincides with line segment $OA$, for $x\in[0,m_\mu]$, and with line segment $AC$, for $x\in[m_\mu,B]$.
Hence,  equality in \eqref{eq:weaktc-3}  holds only if,
	\begin{equation}  \label{eq:weaktc-4}
		\psi^{\prime}_{\mu}\left(z \right) =
		\begin{cases}
			\frac{\psi_{\mu}(m_{\mu})}{m_{\mu}},  & z \in (0, m_{\mu})\\
			\frac{1- \psi_{\mu}(m_{\mu})}{1- m_{\mu}}, & z \in [m_{\mu}, 1).
		\end{cases}
	\end{equation}
By \eqref{eq:regdistortion}, $\psi^{\prime}_{\mu}(z)= \int_{(z,1]} \frac{1}{s }  \mu(ds)$, and using $\psi_{\mu}\left( m_{\mu} \right)= 1-m_{\mu}$,  \eqref{eq:weaktc-4}  becomes,
	\begin{equation}  \label{eq:weaktc-5}
		\int_{(z,1]} \frac{1}{s }  \mu(ds) =
		\begin{cases}
			\frac{1-m_{\mu}}{m_{\mu}},  & z \in (0, m_{\mu})\\
			\frac{m_{\mu}}{1-m_{\mu}}, & z \in [m_{\mu}, 1).
		\end{cases}
	\end{equation}
Note that $\mu$ satisfies \eqref{eq:weaktc-5} only if $\mu= c\delta_{m_{\mu}}+ (1-c)\delta_{1}$, for some $c \in (0, 1)$. Indeed, otherwise there exists an interval $[l, k] \subset (0, m_{\mu})$  or $[l, k] \subset\ (m_{\mu}, 1)$, such that $\mu([l, k])> 0$. If $[l, k] \subset (0, m_{\mu})$, then for any $z_1 \in (0, l)$ and $z_2 \in (k, m_{\mu})$, we have $\int_{(z_2,1]} \frac{1}{s }  \mu(ds)< \int_{(z_1,1]} \frac{1}{s }  \mu(ds)$, which contradicts \eqref{eq:weaktc-5}.  The case $[l, k] \subset (m_{\mu}, 1)$ is treated similarly.
	
Now with $\mu= c\delta_{m_{\mu}}+ (1-c)\delta_{1}$, using \eqref{eq:weaktc-5} we deduce
	\begin{equation} \label{eq:weaktc-6}
		\begin{cases}
			c+  m_{\mu}(1- c)= 1- m_{\mu}\\
			(1-m_{\mu}) (1-c)    = m_{\mu}
		\end{cases}.
	\end{equation}
Solving  \eqref{eq:weaktc-6} for $c$, we obtain that $\mu= \frac{1- 2m_{\mu}}{1- m_{\mu}} \delta_{m_{\mu}}+ \frac{m_{\mu}}{1-m_{\mu}} \delta_{1}$.
Note that by Lemma~\ref{le:disprop1} and  since $\psi_{\mu}\left( m_{\mu} \right)+ m_{\mu}-1= 0$, we have
\begin{align*}
0< \frac{m_{\mu}}{1- m_{\mu}}= \frac{1- \psi_{\mu}(m_{\mu})}{1- m_{\mu}} \leq \frac{1- m_{\mu}}{1- m_{\mu}}= 1,
\end{align*}
and thus  $a:= \frac{1-m_{\mu}}{m_{\mu}} \geq 1$. Consequently, $\mu= \frac{a-1}{a} \delta_{\frac{1}{a+1}}+ \frac{1}{a} \delta_{1}\in\cP'$.
	
Finally we show that if $\mu\in\cP'$ then \eqref{eq:weaktc-3} becomes equality.
We start by representing  $m_{\mu}$ in terms of $\mu$,
	\begin{align}
			m_{\mu}&=\int_{(0,1)}  z \psi_{\mu}^{'}(z)  \dif z= \int_{(0,1)}  z \int_{(z,1]} \frac{1}{s} \mu(ds)  \dif z= \int_{(0,1)}  z \int_{(0,1]} \1_{\{z<s\}} \frac{1}{s} \mu(ds)  \dif z\\
			&= \int_{(0,1]} \frac{1}{s}  \int_{(0,1)} z  \1_{\{z<s\}}   \dif z \mu(\dif s)= \frac{1}{2}  \int_{(0,1]} s \mu(\dif s). \label{eq:weaktc-6-2}
	\end{align}
Using the form of $\mu\in\cP'$, the last equality implies that $m_{\mu}= \frac{1}{a+1}$. From here, and \eqref{eq:regdistortion}, by direct calculations we deduce
\begin{align*}
		\psi_{\mu}\left( m_{\mu} \right)+ m_{\mu}-1 &=  \int_{(0,m_{\mu}]} \int_{(z,1]} \frac{1}{s} \mu(\dif s)  \dif z + m_{\mu}-1\\
		&=   \int_{(0,1]} \1_{z \leq m_{\mu}} \int_{(0,1]} \1_{\{z<s\}} \frac{1}{s} \mu(\dif s)  \dif z + m_{\mu}-1\\
		&=   \int_{(0,1]} \frac{1}{s} \int_{(0,1]} \1_{z < \min \{ m_{\mu}, s \} }  \dif z  \mu(ds)  + m_{\mu}-1\\
		&= \int_{(0,1]}  \min \left\{ \frac{m_{\mu}}{s}, 1 \right\}  \mu(\dif s) + m_{\mu}-1 \\
		&= \int_{(0,m_{\mu}]}   \mu(\dif s)+ \int_{(m_{\mu}, 1]}  \frac{m_{\mu}}{s}  \mu(\dif s) + m_{\mu}-1 \\
		&= \int_{(0,\frac{1}{a+1}]}   \mu(\dif s)+ \frac{1}{a+1} \int_{(\frac{1}{a+1}, 1]}  \frac{1}{s}  \mu(\dif s) + \frac{1}{a+1}-1 \\
		&= \frac{a-1}{a} + \frac{1}{a+1} \frac{1}{a}  + \frac{1}{a+1}-1 \\
		&= 0.
	\end{align*}
This completes the proof.
\hfill $\square$

\begin{lemma} \label{lemma:nonweak-2}
For any probability measure $\mu\in\cP[0, 1]\setminus \cP'$,
and  such that $\psi_{\mu}$ is absolutely continuous, there exist $a, b, c, d \in \mathbb{R}$ satisfying the following conditions:
	\begin{enumerate}[(i)]
		\item  $a <0 < d$, $a< b < c < b$, $ \ 2(c-b)= d-a$;
		\item  $\frac{(a+d)}{a-d}+ \psi_{\mu}\left(\frac{2(b-a)}{d-a}\right) \leq   \int_{(0,1)}  z \psi_{\mu}^{'}(z)  \dif z$;
		\item  $\frac{2b}{a-d}  \leq  \int_{(0,1)}  z \psi_{\mu}^{'}(z)  \dif z$;
		\item $  \int_{(0,1)}  z \psi_{\mu}^{'}(z)  \dif z < \frac{a}{a-d}$.
	\end{enumerate}
\end{lemma}

\begin{proof}
	Let $m_{\mu}:=  \int_{(0,1)}  z \psi_{\mu}^{'}(z)  \dif z$. Then, using representation \eqref{eq:weaktc-6-2}, and  since $\mu\neq \delta_1 $, we obtain that $0< m_{\mu}< \frac{1}{2}$.
	Define $f: [0,1] \rightarrow  \mathbb{R}$ by $f(z)= \psi_{\mu}\left( z \right)+ z-1$. Lemma~\ref{le:disprop1} implies that $f(\frac{1}{2})= \psi_{\mu}\left( \frac{1}{2} \right)+ \frac{1}{2}-1> 0$, and due to Lemma~\ref{lemma:nonweak-1}, $f(m_{\mu})< 0$. Since $f$ is continuous, by intermediate value theorem, there exits  $z_0 \in (m_{\mu}, \frac{1}{2})$, such that $ f(z_0)= 0$.
	
Take $a, b, c, d$ as follows,
$$
a= -m_{\mu}- z_0, \ b= -m_{\mu}, \ c= 1-m_{\mu},\ d= 2-m_{\mu}-z_0.
$$
	We now verify that $a, b, c, d$ satisfy conditions ($i$)-($iv$).
Clearly,
\begin{align*}
d= 2-m_{\mu}-z_0 >1 > c= 1-m_{\mu} >0 > b= -m_{\mu} > a= -m_{\mu}-z_0, \\
d- c= 1-z_0> 0,  \quad c-b= 1 > 0, \quad  b-a= z_0> 0, \quad   c-b=  \frac{d-a}{2},
\end{align*}
and thus (i) is satisfied.
As far as (ii) is concerned,
\begin{align*}
		\frac{(a+d)}{a-d}+ \psi_{\mu}\left(\frac{2(b-a)}{d-a}\right)= m_{\mu}+z_0 -1 + \psi_{\mu} \left(z_0 \right) \leq m_{\mu}.
\end{align*}
Relations (iii) and (iv) follow directly by substituting the above values of $a,b,c,d$.
\end{proof}

\begin{lemma}  \label{lemma:nonweak-3}
Let  $a, b, c, d \in \mathbb{R}$ satisfy conditions ($i$)-($iv$) in Lemma~\ref{lemma:nonweak-2}. Take  $X \sim U[a, d]$ and construct the filtration $\mathbb{F}= (\sF_t)_{t \in \{0, 1, 2 \}}$ as follows
$$
\sF_0 = \{ \emptyset, \Omega \}, \quad \sF_{1}= \sigma \left( \left\{ X^{-1}\left([a, b) \cup  [c, d] \right) \right\}, \left\{ X^{-1}\left([b, c) \right) \right\} \right), \quad  \sF_{2}= \sigma(X).
$$
Then, on this filtered probability space, for any probability measure $\mu\in\cP[0,1]\setminus\cP'$ and such that $\psi_{\mu}$ is absolutely continuous,
	\begin{align*}
		\rho_{1}^{\psi_{\mu}}(X) \leq 0 \  \Rightarrow \ \rho_{0}^{\psi_{\mu}}(X)  >0.
	\end{align*}
\end{lemma}

\begin{proof}
	By  Lemma~\ref{lemma:nonweak-2} condition ($i$), $ \ 2(c-b)= d-a$, two events in $\sF_{1}$ have same probability, i.e.
	\[
	\mathbb{P}\left(\left\{ X^{-1}\left([a, b) \cup  [c, d] \right) \right\}\right)=  \mathbb{P}\left(\left\{ X^{-1}\left( [b, c) \right)  \right\}\right)= \frac{1}{2}.
	\]
We first show that $\rho_{1}^{\psi_{\mu}}(X) \leq 0 $ on the event $\left\{ X^{-1}\left([a, b) \cup  [c, d] \right) \right\}$.
\begin{align*}
\rho_{1}^{\psi_{\mu}}(X)&= \int_{(0,\frac{2(b-a)}{d-a})}  q^{+}_{1-z}(-X \mid \sF_1)  \psi_{\mu}^{'}(z)  \dif z+ \int_{[\frac{2(b-a)}{d-a},1)}  q^{+}_{1-z}(-X \mid \sF_1)  \psi_{\mu}^{'}(z)  \dif z \\
&= \int_{(0,\frac{2(b-a)}{d-a})}  \left( \frac{(1-z)(d-a)}{2}- (d-c)-b \right) \psi_{\mu}^{'}(z)  \dif z \\
&\quad + \int_{[\frac{2(b-a)}{d-a},1)}  \left(\frac{(1-z)(d-a)}{2}- d \right)  \psi_{\mu}^{'}(z)  \dif z\\
&= \int_{(0,\frac{2(b-a)}{d-a})}  \left( -\frac{(d-a)z}{2}- a \right) \psi_{\mu}^{'}(z)  \dif z+ \int_{[\frac{2(b-a)}{d-a},1)}  \left(-\frac{(d-a)z}{2}- \frac{(a+d)}{2} \right)  \psi_{\mu}^{'}(z)  \dif z\\
&= - a \psi_{\mu}\left(\frac{2(b-a)}{d-a}\right)-\frac{(d-a)}{2} \int_{(0,\frac{2(b-a)}{d-a})}  z \psi_{\mu}^{'}(z)  \dif z\\ 
& - \frac{(a+d)}{2} \left[1-\psi_{\mu}\left(\frac{2(b-a)}{d-a}\right)\right]- \frac{(d-a)}{2}\int_{[\frac{2(b-a)}{d-a},1)}  z \psi_{\mu}^{'}(z)  \dif z\\
&= -\frac{(a+d)}{2}+ \frac{(d-a)}{2} \psi_{\mu}\left(\frac{2(b-a)}{d-a}\right)- \frac{(d-a)}{2}\int_{(0,1)}  z \psi_{\mu}^{'}(z)  \dif z\\
&\leq 0.
\end{align*}	
Next, we show that $\rho_{1}^{\psi_{\mu}}(X) \leq 0 $ on the event $\left\{ X^{-1}\left( [b, c) \right) \right\}$. Indeed,
\begin{align*}
\rho_{1}^{\psi_{\mu}}(X)&= \int_{(0,1)}  q^{+}_{1-z}(-X \mid \sF_1)  \psi_{\mu}^{'}(z)  \dif z \\
&= \int_{(0,1)}  \left( \frac{(1-z)(d-a)}{2}-c  \right) \psi_{\mu}^{'}(z)  \dif z \\
&= -b-  \frac{(d-a)}{2} \int_{(0,1)}  z \psi_{\mu}^{'}(z)  \dif z\\
&\leq 0.
\end{align*}

	Finally we prove that $\rho_{0}^{\psi_{\mu}}(X) > 0 $,
	\begin{align*}
		\rho^{\psi_{\mu}}_{0}(X)=& \int_{(0,1)}   q^{+}_{1-z}(-X )  \psi_{\mu}^{'}(z)  \dif z\\
		=& \int_{(0,1)}   \left( -z(d-a)-a  \right) \psi_{\mu}^{'}(z)  \dif z\\
		=& -a- (d-a)\int_{(0,1)}  z \psi_{\mu}^{'}(z)  \dif z\\
		>&0.
	\end{align*}
This concludes the proof.
\end{proof}

We remark that Lemma~\ref{lemma:nonweak-1},  Lemma~\ref{lemma:nonweak-2}  and  Lemma~\ref{lemma:nonweak-3} show that there exists a filtered probability space,  such that $\rho^\psi$ is not weakly acceptance time consistent, regardless of $\mu\in \cP[0,1]\setminus\cP'$ such that $\psi_{\mu}$ is absolutely continuous. Next we will show that this also holds true for any $\mu\in\cP'$, except $\mu=\delta_1$.
We start with a representation of  $\rho^{\psi_{\mu}}$ generated by $\mu\in\cP'$ in terms of a weighted average of conditional AV@R and regular conditional expectation.
\begin{lemma}  \label{0329 lemma-4.22}
For any $t \in \mathcal{T}$, $X \in L^\infty$, $\omega \in \Omega \backslash N^X$, $a \geq 1$, and corresponding  $\mu\in\cP'$, we have that
\[
\rho^{\psi_{\mu}}_t(X)(\omega)= \frac{a-1}{a} \avar_{\frac{1}{a+1}}(X \mid \sF_t)(\omega)- \frac{1}{a} \bE(X \mid \sF_t)(\omega).
\]
\end{lemma}

\begin{proof}
	Recall that if $\mu\in\cP'$, then $\psi_\mu$ is piece-wise linear (see the proof of Lemma~\ref{lemma:nonweak-1}).
By Lemma~\ref{le:altdwar},  combined with \eqref{eq:weaktc-4}, and the facts that $\psi_{\mu}\left( m_{\mu} \right)=1- m_{\mu}$ and  $m_\mu=\frac{1}{1+a}$, we obtain
\begin{align*}
\rho^{\psi_{\mu}}_t(X)(\omega)& = -\int_{(0,1)}   q^{+}_{z}(X \mid \sF_t)(\omega) \psi_{\mu, +}^{'}(z)  \dif z \\
& = -a \int_{(0,\frac{1}{a+1})}    q^{+}_{z}(X \mid \sF_t)(\omega)  \dif z- \frac{1}{a} \int_{[\frac{1}{a+1}, 1)}   q^{+}_{z}(X \mid \sF_t)(\omega)   \dif z \\
& = -a \int_{(0,\frac{1}{a+1})}    q^{+}_{z}(X \mid \sF_t)(\omega)  \dif z- \frac{1}{a} \int_{(0, 1)}   q^{+}_{z}(X \mid \sF_t)(\omega)   \dif z\\
& \qquad  + \frac{1}{a} \int_{(0, \frac{1}{a+1})}   q^{+}_{z}(X \mid \sF_t)(\omega)   \dif z.
\end{align*}
In view of \eqref{eq:avar}, we continue,
	\begin{align*}
		\rho^{\psi_{\mu}}_t(X)(\omega)&= (a- \frac{1}{a}) \int_{(0, \frac{1}{a+1})}   q^{+}_{z}(X \mid \sF_t)(\omega)   \dif z- \frac{1}{a} \bE(-X \mid \sF_t)(\omega) \\
		&= \frac{a-1}{a} \avar_{\frac{1}{a+1}}(X \mid \sF_t)(\omega)- \frac{1}{a} \bE(X \mid \sF_t)(\omega).
	\end{align*}
The proof is complete.
\end{proof}

\begin{lemma} \label{lemma:nonweak-4}
	For any $a> 1$, take $X$ and the filtration $\mathbb{F}= (\sF_t)_{t \in \{0, 1, 2 \}}$ as follows,
$$
X = 2\cdot \1_{\Omega_1} + \1_{\Omega_2}-\frac{a+2}{a} \cdot\1_{\Omega_3} - \frac{2a+4}{a}\cdot\1_{\Omega_4},
$$
where $\Omega_j, \ j=1,2,3,4$, are such that
$$
\mathbb{P}(\Omega_1)= \mathbb{P}(\Omega_2) = \frac{1}{2}- \frac{1}{4(a+1)}, \quad
 \mathbb{P}(\Omega_3)= \mathbb{P}(\Omega_4)= \frac{1}{4(a+1)},
$$
and put
$$
\sF_0 = \{ \emptyset, \Omega \}, \  \sF_{1}= \sigma \left( \left\{ X^{-1}(2) \cup  X^{-1}\left(-\frac{2a+4}{a}\right) \right\}, \left\{  X^{-1}(1) \cup  X^{-1}\left(-\frac{a+2}{a}\right) \right\} \right),
$$
and $\sF_{2}=\sF=  \sigma(X)$.
Then, for $\mu\in\cP'$ corresponding to the above $a>1$ it holds
\begin{align*}
		\rho_{1}^{\psi_{\mu}}(X) \leq 0  \Rightarrow \rho_{0}^{\psi_{\mu}}(X)  >0.
\end{align*}
\end{lemma}

\begin{proof} By  Lemma~\ref{0329 lemma-4.22}, on the event $\left\{ X^{-1}(2) \cup  X^{-1}\left(-\frac{2a+4}{a}\right) \right\}$,
	\begin{align*}
		\rho^{\psi_{\mu}}_{1}(X)&= \frac{a-1}{a} \avar_{\frac{1}{a+1}}\left(X \mid \sF_1\right)- \frac{1}{a} \bE(X \mid \sF_1)\\
		&=   \frac{a-1}{a} \Big(  \frac{a+2}{a} -  1\Big)- \frac{1}{a}\Big( -\frac{2a+4}{a} \frac{1}{2(a+1)}+ 2\big(1- \frac{1}{2(a+1)}\big) \Big) \\
		&= 0.
	\end{align*}
Similarly, on the event $\left\{  X^{-1}(1) \cup  X^{-1}\left(-\frac{a+2}{a}\right) \right\}$,
	$$
		\rho^{\psi_{\mu}}_{1}(X) 
		= \frac{a-1}{a}\left( \frac{a+2}{2a} - \frac{1}{2}  \right)  - \frac{1}{a} \left( -\frac{a+2}{a} \frac{1}{2(a+1)}+ \left(1- \frac{1}{2(a+1)}\right) \right) = 0
	$$
	On the other hand,
	\begin{align*}
		\rho^{\psi_{\mu}}_{0}(X)
		&= \frac{a-1}{a}\left(\frac{a+2}{2a} + \frac{a+2}{4a} - \frac{1}{2}\right)-
		\frac{1}{a}\left[\left(-\frac{2a+4}{a}-\frac{a+2}{a}\right)\frac{1}{4(a+1)}+3 \left(\frac{1}{2}- \frac{1}{4(a+1)}\right)  \right]\\
		&= \frac{a-1}{4a} >0,
	\end{align*}
and the proof is complete.
\end{proof}

Combining Lemma~\ref{lemma:nonweak-1},  Lemma~\ref{lemma:nonweak-2},  Lemma~\ref{lemma:nonweak-3} and Lemma~\ref{lemma:nonweak-4} the proof of Theorem~\ref{thm:nonweak} is immediate.
Here, we remark that part of the constructive proof of Theorem~\ref{thm:nonweak} was partially inspired by a counterexample in \cite{ArtznerDelbaenEberHeathKu2007}.

Finally, since the weak acceptance time consistency is implied by the strong time consistency and middle acceptance time consistency (see \cite{BCP2017}), Theorem~\ref{thm:nonweak} also shows that $\rho^{\psi_{\mu}}$ generally speaking is neither strongly time consistent nor middle acceptance time consistent, with exception of the trivial case $\psi(x)=x$.

\subsection{Time consistency of  $\alpha^\Psi$} \label{sec:timeConsistencyDCAI}

In \cite{ChernyMadan2009} the authors introduce a class of performance measures, called Coherent Acceptability Indices, as functions  $\alpha:L^\infty\to [0,+\infty]$, that are monotone increasing, scale invariant and quasi-concave. Under some technical continuity assumptions, it can be shown that any CAI $\alpha$ can be characterized by a family of CRMs $\set{\bar\rho^x, \ x\geq 0}$. In particular,  each $\bar\rho^x$ can be generated by a distortion function. Applications of these types of performance measures go beyond the risk management, and were successfully done, for example, in portfolio management and pricing of derivatives (cf. \cite{MadanSchoutens2016} and references therein).

The time consistency of acceptability indices, first studied in \cite{BCZ2010}, is a delicate issue. It is can be shown (cf. \cite{BCP2017,BCP2017a}), that a dynamic (coherent) acceptability index, in particular $\alpha^\Psi$, is never strongly time consistent, and in contrast to dynamic risk measures, it can not satisfy a the `recursive' property in principle. In \cite{BCZ2010}, the authors introduce a time consistency property specific to DCAI, called the  weak time consistency in the current manuscript. See also \cite{BCC2014,Bion-Nadal2004,Bion-Nadal2008,RosazzaGianinSgarra2012,BiaginiBion-Nadal2012} for further theoretical developments on acceptability indices. Thanks to robust type representations such as \eqref{def:DCAI}, the time consistency of $\rho^{\psi_x}, x>0$, can  be transferred to appropriate time consistency of $\alpha^\Psi$; see \cite{BCP2017,BCP2017a}.

\begin{proposition}
If $\Psi =(\psi_x)_{x >0}$ is an increasing family of regular distortion functions,  then $\alpha^{\Psi}$ is weakly rejection time consistent.
\end{proposition}
The proof follows from \cite[Proposition~12]{BCP2017}, and the fact that $\rho^{\psi_x}, x>0$, are weakly rejection time consistent, as proved in Section~\ref{sec:weak-rejection-rho}.

\begin{proposition}
  	Let $\Psi =(\psi_x)_{x >0}$ be an increasing and right continuous  family of distortion functions and for each $x >0$,  $\psi_x$ is generated by a probability measure $\mu\in\cP[0,1]\setminus\set{\delta_1}$ and  such that $\psi_x$ is absolutely continuous, then  $\alpha^{\Psi}$ is not weakly acceptance time consistent on some filtered probability space.
\end{proposition}
\begin{proof} It can be proved (cf. \cite[Theorem~4.8]{BCZ2010}) that
$$
 \rho_t^{\psi_x}(X)(\omega)= -\inf \{ c \in \mathbb{R} \mid \alpha^{\Psi}_t(X -c)(\omega) \leq x \}.
$$
Then, in view of \cite[Proposition~13]{BCP2017},  if $\alpha^\Psi$ is weakly acceptance time consistent, then  $\rho^{\psi_x}, x>0$ is also weakly time consistent, which contradicts Theorem~\ref{thm:nonweak}. The proof is complete.
\end{proof}

\begin{appendix}

\section{Auxiliary definitions and results}\label{append:defs}

A function $f:\cT\times L^\infty\times\Omega \to \bR$ is said to be:
\begin{itemize}
	\item \textit{Adapted}, if $f_t(X)$ is $\sF_t$-measurable;
	\item \textit{Normalized}, if $f_t(0) = 0$;
	\item \textit{Local}, if $\1_{A} f_{t}(X)= \1_{A} f_{t}\left(\1_{A} X\right)$;
	\item \textit{Cash Additive}, if $f_t(X+ m) = f_t(X)- m$;
	\item \textit{Monotone increasing}, if $X \leq Y$ implies $f_t(X) \leq f_t(Y)$;
	\item \textit{Monotone decreasing}, if $-f$ is monotone increasing;
	\item \textit{Sub-additive}, if $f_t( X + Y) \leq f_t(X)+ f_t(Y)$;
	\item \textit{Positive Homogeneous}, if $f_t(\gamma X) = \gamma f_t(X)$;
	\item \textit{Quasi-concave}, if $f_{t}(X) \geq n$ and  $f_{t}\left(Y\right) \geq n$,  then  $f_{t}\left(\lambda X+ (1-\lambda) Y\right) \geq n$;
	\item \textit{Scale Invariance}, if $f_{t}(\beta X)= f_{t}(X)$;
	\item \textit{Law-invariant}, if the value of $f_t(X)$ depends only on the conditional distribution of $X$, namely for any $X,Y\in L^\infty$ such that $\bP(X\in A|\sF_t)=\bP(Y\in A|\sF_t)$, for any $A\in\cB(\bP)$, we have that $f_t(X)=f_t(Y)$,
\end{itemize}
for $t \in \mathcal{T}$, $X, Y \in L^\infty$, $A \in \sF_{t}$, $m \in L^\infty_t$, $n, \gamma \in L^\infty_{t, +}$, $\lambda \in L^{\infty}_t$, $0 \leq \lambda \leq 1$, $\beta\in L_t^\infty, \ \beta>0$.

\subsection{Conditional quantiles, $\var$ and $\avar$}

The conditional versions of qualtiles, value at risk ($\var$) and average value at risk ($\avar$), are defined naturally by using `probabilistic conditioning' of the corresponding regular (or static) notions. Most of the properties are also expected to hold after `conditioning'.  While morally this is true, the difficulties are hidden in technical details related to measurability and well-definiteness of these quantities. Here, we define these objects, and state some of their properties, while the details are deferred to the Supplement~\ref{sec:supp}. 

\begin{definition}  \label{def:condquantile} For any $ \alpha \in (0,1)$, the conditional upper and lower $\alpha$-quantile of $X\in L^\infty$ with respect to $\sigma$-field $\sF_t$ are defined as
	\begin{align*}
		q_{\alpha}^{+}(X \mid \sF_t) = \esssup \{m \in L^{\infty}_t \mid \mathbb{P}(X \leq m\mid\sF_t) \leq \alpha \},\\
		q_{\alpha}^{-}(X \mid \sF_t) =\essinf \{m \in L^{\infty}_t \mid  \mathbb{P}(X \leq m\mid\sF_t) \geq \alpha \}.
	\end{align*}
\end{definition}

It can be shown that the upper and lower quantile admit the representations

	\begin{align*}
		q_{\alpha}^{+}(X \mid \sF_t)(\omega)&= \sup \{x \in \mathbb{R} \mid \bP(X \leq x\mid\sF_t)(\omega) \leq \alpha \}  \\
		&= \sup \{x \in \mathbb{R} \mid \bP(X < x\mid\sF_t)(\omega) \leq \alpha \} \\
		q_{\alpha}^{-}(X \mid \sF_t)(\omega) &= \inf \{x \in \mathbb{R} \mid \bP(X \leq x\mid\sF_t)(\omega) \geq \alpha \} \\
		& = \inf \{x \in \mathbb{R} \mid \bP(X < x\mid\sF_t)(\omega) \geq \alpha \},  \\
		q^{+}_{\alpha}(X \mid \sF_t)(\omega) & = - q^{-}_{1-\alpha}(-X \mid \sF_t)(\omega).
\end{align*}

The conditional $\var$,  similar to its static counterpart, is defined in terms of conditional $\alpha$-quantile function.

\begin{definition} For fixed $ \alpha \in (0,1)$ and  $X \in L^\infty$, conditional $\var$ at level $\alpha$ with respect to $\sigma$-field $\sF_t$  is defined as,
	\begin{align}
		\var_{\alpha}(X \mid \sF_t) := -q_{\alpha}^{+}(X \mid \sF_t)= \essinf \{m \in L^{\infty}_t \mid  \mathbb{P}(X+ m < 0 \mid \sF_t) \leq \alpha \}.
	\end{align}	
\end{definition}
From financial point of view, $\var_{\alpha}(X|\sF_t)$ can be viewed as the smallest amount of capital which, if added to the position $X$ at time $t$, implies that (conditional) probability that the secured position $X+m$ will yield losses is below the level $\alpha$. Clearly, the $\var$ does not capture the size or distribution of the losses beyond the quantile. To overcome this, the notion of average value at risk is introduced.

\begin{definition} The conditional Average Value at Risk at level $\alpha \in (0, 1]$ of a position $ X \in L^\infty$ is given by
	\begin{align} \label{eq:avar}
		\avar_{\alpha}(X\mid \sF_t):= \frac{1}{\alpha} \int_{(0, \alpha)} \var_{z}(X \mid \sF_t) \dif z= -\frac{1}{\alpha} \int_{(0, \alpha)} q_{z}^{+}(X \mid \sF_t) \dif z.
	\end{align}
\end{definition}
Note that $ q_{z}^{+}(X \mid \sF_t)$ is monotone increasing with respect to $z$. Since monotonicity implies Borel measurability, $z \mapsto  q_{z}^{+}(X \mid \sF_t)$ is Borel measurable. Thus, the integral in \eqref{eq:avar} is well-defined.

\subsection{Notions of Time Consistency}
In \cite{BCP2017,BCP2017a} the authors introduced a general framework for studying time consistency of local and monotone (increasing) functions $f:\cT\times L^\infty\times\Omega \to \bR$.
In particular, the notion of weakly/middle and rejection or acceptance time consistency, super/sub-martingale time consistency, and strong time consistency were introduced. The time consistency property and its economic interpretation stays at the foundation of the inter-temporal decision making theory. It is beyond the scope of this work to discuss the subtle differences between different forms of time consistency, and we refer the reader to the survey \cite{BCP2017} and references therein. Here, for the sake of completeness, but also to avoid confusions between acceptance and rejection time consistency depending on monotonicity, we simply list some of these notions relevant to the objects we study.

\smallskip
\noindent
A DCRM $\rho$ is:
\begin{itemize}
	\item \textit{strong} time consistent if $\rho_s(X)=\rho_s(Y) \Rightarrow \rho_t(X) = \rho_t(Y)$;
	\item \textit{sub-martingale} time consistent if
$\rho_{t}(X) \geq \mathbb{E}\left[\rho_{s}(X) \mid \sF_{t}\right]$;

\item \textit{ middle  rejection} time consistent if, $\rho_{s}(X) \geq \rho_{s}(Y)  \Rightarrow \rho_{t}(X) \geq \rho_{t}(Y)$, for any $Y \in L^\infty_s$;

\item \textit{ super-martingale} time consistent if
$\rho_{t}(X) \leq \mathbb{E}\left[\rho_{s}(X) \mid \sF_{t}\right]$;

\item \textit{weak acceptance} time consistent if $\rho_{s}(X) \leq 0  \Rightarrow \rho_{t}(X)  \leq 0,$
\end{itemize}
for any $X,Y \in L^\infty$, $s, t \in \mathcal{T}$, such that $s > t $.

\bigskip

\noindent
A DCAI $\alpha$ is:
\begin{itemize}
	\item \textit{weakly rejection} time consistent if,
$\alpha_{s}\left(X \right) \leq m_t    \Rightarrow  \alpha_{t}\left(X\right) \leq m_t$;

\item \textit{ weakly acceptance} time consistent if,
$\alpha_{s}\left(X \right) \geq m_t    \Rightarrow  \alpha_{t}\left(X\right) \geq m_t,$

\end{itemize}
for any  $X \in L^\infty$, $s, t \in \mathcal{T}$ with $s>t$ and non-negative  $ m_t \in \sF_{t}$.

\end{appendix}

\newpage

\bibliographystyle{alpha}

\newcommand{\etalchar}[1]{$^{#1}$}

\newpage 

\section{Part II: Technical Supplement}\label{sec:supp}

Here, we present detailed proofs of some technical results related to Dynamic Coherent Risk Measures (DCRMs) and Dynamic Acceptability Indices (DCAIs) generated by distortion functions, alongside some properties of conditional quantiles, conditional value at risk ($\var$) and conditional average value at risk ($\avar$).   This part can be viewed as a self-contained supplement to the first part of this manuscript. We will follow the notations and definitions as in the first part without repeating many of them here.

Let $T$ be a fixed and finite time horizon, and let $\mathcal{T}:= \{0, 1, ..., T \}$. We consider a filtered probability space $(\Omega, \sF, \{\sF_t\}_{t \in \mathcal{T}} , \mathbb{P})$, with $\sF_0 = \{ \emptyset, \Omega \}$ and $\sF = \sF_T$. Throughout, we will use the notations $L^\infty:= L^\infty(\Omega, \sF, \mathbb{P})$, $L^\infty_t:= L^\infty(\Omega, \sF_t, \mathbb{P})$, $t\in\cT$, and  $L^\infty_{t, +}$ the set of all non-negative random variables in $L^\infty_t$, for $t\in\cT$. As usual, all equalities and inequalities will be understood in $\mathbb{P}$-almost surely sense unless otherwise stated.    The set of all probability measures on $[0,1]$ is denoted by $\cP([0,1])$.  

We recall that for any real-valued random variable $X : (\Omega, \sF) \rightarrow (\mathbb{R}, \mathcal{B}(\mathbb{R}))$ there exists a \textit{regular conditional distribution} of $X$ given the $\sigma$-algebra $\sF_t$. That is,  there exits a null set $N^X_t \in \sF_t$, such that for any $ \omega \in \Omega \backslash N^X_t$ and any $B\in \mathcal{B}(\mathbb{R})$, we have that $\mathbb{P}(X \in B | \sF_t)(\omega)$ is a distribution function (cf. \cite[Theorem 8.29]{Klenke2013}). In what follows, conditional probabilities will be understood in this sense, and since the set $N^X:= {\cup}_{t \in \mathcal{T}} N^X_t$ is also a null set, and we will use it conveniently instead of $N_t$, for all $t\in\cT$.

We start with a result on distortion functions. 

\begin{lemma} \label{le:disprop1-1} 
	For any concave distortion $\psi$ except the identity function, we have that $\psi(x)> x$ for any $x \in (0,1)$. Moreover, 	any concave distortion $\psi$ is continuous on $(0, 1]$.
\end{lemma}
\begin{proof}
	Since $\psi$ is concave on $[0, 1]$, then $\psi(x) \geq x$ for any $x \in [0,1]$. Otherwise suppose there exists $x' \in [0,1]$ such that $\psi(x') < x'$. Take $x= 1$, $y= 0$ and $\alpha= x'$,  then we get, 
	\[ 
	\psi(\alpha x +(1-\alpha)y)=  \psi(x') < x' = \alpha \psi(x) +(1-\alpha) \psi(y), 
	\]
	which contradicts that $\psi$ is concave on $[0, 1]$. 
	
	Since $\psi$ cannot be identity function, thus there exists  $m \in (0, 1)$, such that $\psi(m) > m$. 
	
	Now we prove the lemma by contradiction. Suppose there exists $z \in (0, 1)$, such that  $\psi(z) \leq z$. Since $\psi(x) \geq x$ for any $x \in [0,1]$, then  $\psi(z) = z$.   First assume $m < z$. We will show that $\psi$ is not concave on $[0, 1]$, i.e.
	\begin{align}  \label{eq:notconcave}
		\exists \ x, y  \in [0, 1], \ \exists \ \alpha \in [0, 1],  \  \text{such that} \  \psi(\alpha x +(1-\alpha)y) < \alpha \psi(x) +(1-\alpha) \psi(y).
	\end{align}	
	Take $x= m$, $y= 1$, $\alpha= \frac{1-z}{1-m}$. Then we obtain 
	\begin{align*} 
		\psi(\alpha x +(1-\alpha)y) =  \psi \left(\frac{1-z}{1-m} m + \frac{z-m}{1-m} \right) =\psi(z) = z. 
	\end{align*}
	Since $z= \frac{1-z}{1-m} m + \frac{z-m}{1-m}$, it follows that, 
	\begin{align*} 
		\psi(\alpha x +(1-\alpha)y)= \frac{(1-z)m}{1-m} + \frac{z-m}{1-m}< \frac{(1-z)\psi(m)}{1-m}  + \frac{z-m}{1-m} = \alpha \psi(x) +(1-\alpha) \psi(y). 
	\end{align*} 
	This leads to \eqref{eq:notconcave}, and thus $\psi$ is not concave. The case $m > z$, is treated similarly, take $x= 0$, $y= m$, $\alpha= 1-\frac{z}{m}$, and obtain \eqref{eq:notconcave}.

	To prove the second part, we show that $\psi$ is continuous on $(0, 1)$.  By \cite[Proposition A.4 (a)]{FollmerSchiedBook2004}, concavity of $\psi$ on $[0, 1]$ implies that $\psi$ is locally Lipschitz continuous on $(0, 1)$. 
	So for each point $r \in (0, 1)$, there exists a neighbourhood $O_r$ of $r$, for any $x, y \in O_r$,  there exists $M_r \in \mathbb{R}_+$, such that  $\frac{|\psi(x)-\psi(y)|}{|x-y|} \leq M_r$. Take $y =r$,  then $|\psi(x)-\psi(r)| \leq M_r |x-r| $. Therefore $\psi(x) \rightarrow \psi(r)$ as $x \rightarrow r$, which means that $\psi$ is continuous at $r$. Since $r$ is arbitrary on $(0, 1)$,  $\psi$ is continuous on $(0, 1)$. 
	
	Next for the continuity at point 1, if $\psi$ is identity function, then obviously it is continuous at 1. For other distortion $\psi$,  we use proof by contradiction. Suppose $\psi$ is not continuous at point 1, since $\psi(1)= 1$ and $\psi$ is increasing, assume $\psi(1-)= a< 1$. Take arbitrary $y \in [a, 1)$, since $\psi$ is increasing, $\psi(y) \leq \psi(1-)= a \leq y$. This is in contradiction to  $\psi(y)> y$ for any $y \in (0,1)$.
	This concludes the proof. 
\end{proof}

Next we prove to auxiliary results needed for next section. 
\begin{lemma} \label{le:sub2}
	Let  $A_1, \ldots , A_{n} \in \sF_T $ be a (disjoint) partition of $\Omega$ and put $I:= \{ 1,...,n \}$. For given $t\in\cT$, $\psi \in \Upsilon$, and $\omega \in \Omega \backslash N^X$,  construct the function $Q^{\omega}$  on $\sG:=\sigma (A_{1}, \ldots , A_{n} )$  generated by
	\begin{equation} \label{eq:subQ}
		Q^{\omega}[\bigcup_{k \in J \subseteq I} A_{k}]:= \sum_{k \in J \subseteq I} [\psi(\bP(B_{k}\mid\sF_t)(\omega))- \psi(\bP(B_{k-1}\mid\sF_t)(\omega))],
	\end{equation}
	where $B_{0}:= \emptyset, \  B_{k}:= \bigcup_{j=1}^{k} A_{j}$, $k \in I$. Then, $Q^\omega$ is a probability measure on $(\Omega, \sG)$. 
\end{lemma}
\begin{proof} Clearly 	$Q^{\omega}[\emptyset]= 0$, and 
	\begin{align*}
		Q^{\omega}[\Omega] &= 	Q^{\omega}[\bigcup_{k \in I} A_{k}]= \sum_{k \in I} [\psi(\bP(B_{k}\mid\sF_t)(\omega))- \psi(\bP(B_{k-1}\mid\sF_t)(\omega))]\\ 
		& = \psi(\bP(B_{n}\mid\sF_t)(\omega))- \psi(\bP(B_{0}\mid\sF_t)(\omega))= 1. 
	\end{align*}
	
	Since the $\cG$ is generated by a finite partition, $\sigma$-additivity of $Q^\omega$ is equivalent to finite additivity, which is immediate: 
	$$
	Q^{\omega}[\bigcup_{k \in J \subseteq I} A_{k}]= \sum_{k \in J \subseteq I} [\psi(\bP(B_{k}\mid\sF_t)(\omega))- \psi(\bP(B_{k-1}\mid\sF_t)(\omega))]= \sum_{k \in J \subseteq I} Q^{\omega}[A_{k}].
	$$
	Thus, $Q^\omega$ is a probability measure, and this concludes the proof. 
\end{proof}

\begin{lemma} \label{le:ineqDCRM}
	Given $\psi \in \Upsilon$ and $t \in \mathcal{T}$, for any $\omega \in \Omega \backslash N^X$, let $Q^{\omega}$ be the probability measure given by \eqref{eq:subQ}. Then, for any non-negative random variable $X = \sum_{i=1}^{n} x_{i} \1_{A_{i}}$ we have 
	\begin{equation}\label{eq:QOmega2}
		\sum_{i=1}^{n} -x_{i}Q^{\omega}[A_{i}] \leq  \rho_t^\psi(X)(\omega).
	\end{equation}
\end{lemma}
\begin{proof}			
	First we will show that if  $x_{n} \geq \ldots \geq x_{1} \geq x_0:=0$, then $\rho_t^\psi(X)(\omega)= \sum_{i=1}^{n} -x_{i} Q^{\omega}(A_{i}) $, i.e. \eqref{eq:QOmega2} becomes equality. With $B_k, \ k\in I$, defined in Lemma~\ref{le:sub2}, we notice that 
	$$
	\{ -X > y \} =
	\begin{cases}
		B_{n}, & y < -x_{n}\\
		B_{n-1}, & -x_{n} \leq  y < -x_{{n}-1}\\
		\ldots\\
		B_{1}, & -x_2\leq y < -x_1\\
		B_{0}, & -x_1 \leq  y. 
	\end{cases}
	$$
	This, consequently implies that 
	\begin{align*}
		\rho_t^\psi(X)(\omega) &= \int_{(-\infty, 0)}   [ \psi(\bP(-X > y\mid\sF_t)(\omega)) -1 ]  \dif y\\
		&= \sum_{i=0}^{n-1}   [ \psi(\bP(B_i\mid\sF_t)(\omega)) -1 ] (x_{i+1}- x_{i})\\
		&= -Q^{\omega}(A_{n}) (x_{n}- x_{n-1})- \sum_{i=n-1}^{n} Q^{\omega}(A_{i}) (x_{n-1}- x_{n-2})- \ldots \\
		& \qquad  \qquad - \sum_{i=2}^{n} Q^{\omega}(A_{i}) (x_{2}- x_{1})- \sum_{i=1}^{n} Q^{\omega}(A_{i}) x_{1}  \\
		&= \sum_{i=1}^{n} -x_{i}  Q^{\omega}(A_{i}).
	\end{align*}
	
	Second, we prove  that $\rho_t^\psi(X)(\omega) \geq \sum_{k=1}^{n} -x_{k} Q^{\omega}(A_{k}) $  for arbitrary non-negative  step function $X$.  Note that, any permutation  $\sigma$  of  $\{1, \ldots, n \}$  induces a probability  $Q_{\sigma}^{\omega}$  by applying Lemma~\ref{le:sub2}  to the relabeled partition $ A_{\sigma(1)} \ldots, A_{\sigma(n)}$ .  Let  $\sigma$  be the permutation of $I$ such that  $x_{\sigma(1)} \leq \ldots \leq x_{\sigma(n)}$,  then in view of the above, we have  
	\[
	\rho_t^\psi(X)({\omega})= \sum_{k=1}^{n} -x_{\sigma(k)}  Q^{\omega}_{\sigma}(A_{\sigma(k)}).
	\] 
	Hence, the assertion will follow if we can prove that
	\begin{equation} \label{eq:2-11}
		\sum_{k=1}^{n} -x_{\sigma(k)}  Q^{\omega}_{\sigma}(A_{\sigma(k)}) \geq \sum_{k=1}^{n} -x_{k}  Q^{\omega}(A_{k}).
	\end{equation}
	We claim that \eqref{eq:2-11} will follow if we can prove that 
	\begin{equation} \label{eq:2-12}
		\sum_{k=1}^{n} -x_{\tau(k)}  Q^{\omega}_{\tau}(A_{\tau(k)}) \geq \sum_{k=1}^{n} -x_{k}  Q^{\omega}(A_{k}),
	\end{equation}
	where  $\tau$  is the transposition of arbitrary, but fixed, two neighboring indices  $i$  and  $i+1$ where $ x_{i} \geq x_{i+1}$, such that $ x_{\tau(i)} \leq x_{\tau(i+1)}$; operation $\tau$ only affects the positions $i$  and  $i+1$, with all other terms remaining terms unchanged. Let us show that indeed \eqref{eq:2-12} implies \eqref{eq:2-11}. If the values of $X$ are not arranged in increasing order, then there exists at least one pair of neighboring indices with decreasing order, so that we can apply $m$ finitely many times neighboring swapping operation $\tau$ until we get the permutation $\sigma$, and at each step we relabel $X$. At $j$-th step, we swap two reverse-order
	neighboring indices of $X_{\tau}^{j-1}$ which is the relabeled sequence from step $j-1$. Then, we get that $X^{j}_{\tau}= \sum_{k=1}^{n} x_{\tau(k)}^{j-1} A^{j-1}_{\tau(k)}$, and by \eqref{eq:2-12}, we deduce that 
	\begin{equation} \label{eq:2-13}
		\bE^{Q^{\omega}_j}(-X^{j}_{\tau})= \sum_{k=1}^{n} -x^{j-1}_{\tau(k)}  Q^{\omega}_j(A^{j-1}_{\tau(k)}) \geq   \sum_{k=1}^{n} -x^{j-1}_k  Q^{\omega}_{j-1}(A^{j-1}_k )= \bE^{Q^{\omega}_{j-1}}(-X^{j-1}_{\tau}).
	\end{equation}
	where $Q^{\omega}_j$ and $Q^{\omega}_{j-1}$ correspond to $ Q^{\omega}_{\tau}$ and $Q^{\omega}$ in \eqref{eq:2-12} respectively. Taking $j=m$, we get $X^{m}_{\tau}$ being the same as $X_{\sigma}$ since both are the relabeled $X$ in increasing order. Moreover,  $Q^{\omega}_{\sigma}$ and $Q^{\omega}_{m}$ are the same since both are constructed based on the same sequence of sets, so that
	\begin{equation} \label{eq:2-14}
		\bE^{Q^{\omega}_{\sigma}}(-X_{\sigma})   =  \bE^{Q^{\omega}_m}(-X^m_{\tau}).
	\end{equation}
	By applying \eqref{eq:2-13} inductively, we have
	\begin{equation} \label{eq:2-15}
		\bE^{Q^{\omega}_m}(-X^m_{\tau})  \geq \bE^{Q^{\omega}_{m-1}}(-X^{m-1}_{\tau}) \geq \ldots \geq \bE^{Q^{\omega}_{1}}(-X^{1}_{\tau}) \geq \bE^{Q^{\omega}}(-X).
	\end{equation}
	Combining \eqref{eq:2-14} and \eqref{eq:2-15},  then \eqref{eq:2-11} follows, and thus  \eqref{eq:2-12} implies \eqref{eq:2-11}.

	Finally, it remains to prove \eqref{eq:2-12}. Recall that $\tau(i)= i+1$, $\tau(i+1)= i$ and $\tau(k)= k$ for all other $k$. Thus,
	\begin{align}
		-\sum_{k=1}^{n} x_{\tau(k)}  Q_{\tau}^{\omega}(A_{\tau(k)}) +\sum_{k=1}^{n} x_{k}  Q^{\omega}(A_{k}) 
		&= -x_{\tau(i)} Q_{\tau}^{\omega}\left[A_{\tau(i)}\right]- x_{\tau(i+1)} Q^{\omega}_{\tau}\left[A_{\tau(i+1)}\right] \\
		& \qquad  +x_{i} Q^{\omega}\left[A_i \right]+ x_{i+1} Q^{\omega}\left[A_{i+1}\right] \\
		&= -x_{i+1} Q^{\omega}_{\tau}\left[A_{i+1}\right]- x_{i} Q^{\omega}_{\tau}\left[A_{i}\right] \\
		& \qquad + (x_{i} Q^{\omega}\left[A_i \right]+ x_{i+1} Q^{\omega}\left[A_{i+1}\right]) \\
		&= -x_{i}\left(Q^{\omega}_{\tau}\left[A_{i}\right]-Q^{\omega}\left[A_{i}\right]\right)-x_{i+1}\left(Q^{\omega}_{\tau}\left[A_{i+1}\right]-Q^{\omega}\left[A_{i+1}\right]\right).
	\end{align}
	To compute $Q^{\omega}_{\tau}\left(A_{k}\right)$, let us introduce
	\begin{equation}
		B_{0}^{\tau}:=\emptyset \quad \text { and } \quad B_{k}^{\tau}:=\bigcup_{j=1}^{k} A_{\tau(j)}, \quad k=1, \ldots, n.
	\end{equation}
	We note that $B_{i}^{\tau}= \bigcup_{j=1}^{i-1} A_{j} \bigcup A_{i+1}$ 	and  $B_{k}^{\tau}=B_{k}$,  for  $k \neq i$.  Hence,
	\begin{equation} \label{eq:2-18}
		\begin{split}
			Q_{\tau}^{\omega}\left[A_{i}\right]+ Q^{\omega}_{\tau}\left[A_{i+1}\right] 
			& =Q^{\omega}_{\tau}\left[A_{\tau(i+1)}\right] + Q^{\omega}_{\tau}\left[A_{\tau(i)}\right]\\ 
			& =  \psi(\bP(B_{i+1}^{\tau}\mid\sF_t)(\omega))- \psi(\bP(B_{i}^{\tau}\mid\sF_t)(\omega)) \\
			&\qquad + \psi(\bP(B_{i}^{\tau}\mid\sF_t)(\omega))- \psi(\bP(B_{i-1}^{\tau}\mid\sF_t)(\omega)) \\
			&=  \psi(\bP(B_{i+1}^{\tau}\mid\sF_t)(\omega))- \psi(\bP(B_{i-1}^{\tau}\mid\sF_t)(\omega)) \\
			&=  \psi(\bP(B_{i+1}\mid\sF_t)(\omega))- \psi(\bP(B_{i-1}\mid\sF_t)(\omega)) \\
			&= Q^{\omega}\left[A_{i}\right]+Q^{\omega}\left[A_{i+1}\right].
		\end{split}
	\end{equation}
	Moreover,  $B_{i}^{\tau} \cap B_{i}=B_{i-1}$,  $B_{i}^{\tau} \cup B_{i}=B_{i+1}$ and we have,
	\begin{equation} \label{eq:2-19}
		\begin{split}
			\bP(B_{i}\mid\sF_t)(\omega) &- \bP(B_{i-1}\mid\sF_t)(\omega) 
			=  \bP(-X > -x_{i+1}\mid\sF_t)(\omega) - \bP(-X > -x_{i}\mid\sF_t)(\omega)\\
			& =  \bP(-X > -x_{i+2}\mid\sF_t)(\omega) - \bP( \{ -X > -x_{i} \} \cup \{-x_{i+1}  \geq -X > -x_{i+2} \} \mid\sF_t)(\omega)\\
			& =  \bP(B_{i+1}\mid\sF_t)(\omega)- \bP(B_{i}^{\tau}\mid\sF_t)(\omega).
		\end{split}
	\end{equation}
	Since $\psi$ is concave, we deduce 
	\[
	\frac{\psi(\bP(B_{i+1}\mid\sF_t)(\omega))- \psi(\bP(B_{i}^{\tau}\mid\sF_t)(\omega))}{\bP(B_{i+1}\mid\sF_t)(\omega)- \bP(B_{i}^{\tau}\mid\sF_t)(\omega)} \leq \frac{\psi(\bP(B_{i}\mid\sF_t)(\omega))- \psi(\bP(B_{i-1}\mid\sF_t)(\omega))}{\bP(B_{i}\mid\sF_t)(\omega)- \bP(B_{i-1}\mid\sF_t)(\omega)},
	\]
	and invoking \eqref{eq:2-19}, we obtain 
	\[
	\psi(\bP(B_{i+1}\mid\sF_t)(\omega))- \psi(\bP(B_{i}^{\tau}\mid\sF_t)(\omega)) \leq \psi(\bP(B_{i}\mid\sF_t)(\omega))- \psi(\bP(B_{i-1}\mid\sF_t)(\omega)).
	\]
	Thus,
	\begin{equation} \label{eq:2-20}
		\begin{split}
			Q^{\omega}\left[A_{i+1}\right]&= \psi(\bP(B_{i+1}\mid\sF_t)(\omega))- \psi(\bP(B_{i}\mid\sF_t)(\omega))\\  &\leq   \psi(\bP(B_{i}^{\tau}\mid\sF_t)(\omega))- \psi(\bP(B_{i-1}^{\tau}\mid\sF_t)(\omega))\\  &= Q^{\omega}_{\tau}\left[A_{\tau(i)}\right]\\ &=Q^{\omega}_{\tau}\left[A_{i+1}\right].
		\end{split}
	\end{equation}
	By  \eqref{eq:2-18} and \eqref{eq:2-20}, we have $Q^{\omega}\left[A_{i+1}\right] \leq Q^{\omega}_{\tau}\left[A_{i+1}\right] \hspace{2mm} $ and $Q^{\omega}\left[A_{i}\right] \geq Q^{\omega}_{\tau}\left[A_{i}\right] $. Moreover, 
	\[
	Q^{\omega}\left[A_{i}\right]- Q^{\omega}_{\tau}\left[A_{i}\right]= Q^{\omega}_{\tau}\left[A_{i+1}\right]- Q^{\omega}\left[A_{i+1}\right]. 
	\]
	Due to our assumption  $x_{i}\geq x_{i+1}$, and thus 
	\begin{align*}
		-\sum_{k=1}^{n} x_{\tau(k)}  Q^{\omega}_{\tau}(A_{\tau(k)}) +\sum_{k=1}^{n} x_{k}  Q^{\omega}(A_{k})&= -x_{i}\left(Q^{\omega}_{\tau}\left[A_{i}\right]-Q^{\omega}\left[A_{i}\right]\right)  -x_{i+1}\left(Q^{\omega}_{\tau}\left[A_{i+1}\right]-Q^{\omega}\left[A_{i+1}\right]\right)\\
		&= (x_{i}- x_{i+1}) \left(Q^{\omega}_{\tau}\left[A_{i+1}\right]-Q^{\omega}\left[A_{i+1}\right]\right)
		\geq 0.
	\end{align*} 
	Hence \eqref{eq:2-12} holds true. This completes the proof.
\end{proof}	

For the sake of completeness, we also recall here that a function $f:\cT\times L^\infty\times\Omega \to \bR$ is said to be: 
\begin{itemize}
	\item \textit{Adapted}, if $f_t(X)$ is $\sF_t$-measurable; 
	\item \textit{Normalized}, if $f_t(0) = 0$;   
	\item \textit{Local}, if $\1_{A} f_{t}(X)= \1_{A} f_{t}\left(\1_{A} X\right)$; 
	\item \textit{Cash Additive}, if $f_t(X+ m) = f_t(X)- m$; 
	\item \textit{Monotone increasing}, if $X \leq Y$ implies $f_t(X) \leq f_t(Y)$; 
	\item \textit{Monotone decreasing}, if $-f$ is monotone increasing; 
	\item \textit{Sub-additive}, if $f_t( X + Y) \leq f_t(X)+ f_t(Y)$; 
	\item \textit{Positive Homogeneous}, if $f_t(\gamma X) = \gamma f_t(X)$; 
	\item \textit{Quasi-concave}, if $f_{t}(X) \geq n$ and  $f_{t}\left(Y\right) \geq n$,  then  $f_{t}\left(\lambda X+ (1-\lambda) Y\right) \geq n$;  
	\item \textit{Scale Invariance}, if $f_{t}(\beta X)= f_{t}(X)$;
	\item \textit{Law-invariant}, if the value of $f_t(X)$ depends only on the conditional distribution of $X$, namely for any $X,Y\in L^\infty$ such that $\bP(X\in A|\sF_t)=\bP(Y\in A|\sF_t)$, for any $A\in\cB(\bP)$, we have that $f_t(X)=f_t(Y)$,
\end{itemize}
for $t \in \mathcal{T}$, $X, Y \in L^\infty$, $A \in \sF_{t}$, $m \in L^\infty_t$, $n, \gamma \in L^\infty_{t, +}$, $\lambda \in L^{\infty}_t$, $0 \leq \lambda \leq 1$, $\beta\in L_t^\infty, \ \beta>0$.

\subsection{On DCRMs}

For any $\psi,  X \in L^\infty$  and $t \in \mathcal{T}$, we define 
\begin{align} \label{eq:ChoquetDCRM-supp}
	\rho_t^\psi(X):= \int_{[0,\infty)}   \psi \left(\bP(-X > y \mid \sF_t)\right)   dy+ \int_{(-\infty,0)}  \left[ \psi(\bP(-X > y \mid \sF_t)) -1 \right]  \dif y. 
\end{align} 
\begin{definition}\label{def:DCRM-supp} 
	A \textit{Dynamic Coherent Risk Measure} (DCRM) is a function $\rho: \mathcal{T} \times L^\infty \times \Omega \rightarrow \mathbb{R}$ that is adapted, normalized,  local, cash-additive, monotone decreasing, sub-additive and positive homogeneous.
\end{definition}

\begin{proposition} \label{prop:DCRM-supp} 
	The mapping $\rho^\psi$ is a DCRM. 
\end{proposition}
\begin{proof} 	
	First we show that \eqref{eq:ChoquetDCRM-supp} is well-defined, that is:
	\begin{enumerate}[(i)]
		\item For any $t \in \mathcal{T}$, $X \in L^\infty$ and $y  \in \mathbb{R}$, the mapping $y \mapsto \psi(\bP(-X > y \mid \sF_t)(\omega))$ is Borel measurable, and thus the integrals in \eqref{eq:ChoquetDCRM-supp} are well-defined; 
		\item $\rho_t^\psi(X) \in L^\infty$, for all $t \in \mathcal{T}$, $X \in L^\infty$.
	\end{enumerate}
	Since for any $ \omega \in \Omega \backslash N^{X}$, $\bP(-X > y \mid \sF_t)(\omega)$ is a distribution function, then it is Borel measurable with respect to $y$. Consequently, $\psi$ being continuous, preserves measurability, and (i) is proved.  Next take $a$ and $b$ to be the essential infimum and the essential supremum of $X$, respectively, and without loss of generality, assume that $a < 0 <  b$. Then,  
	\begin{align*}
		\rho_t^\psi(X)(\omega) &= \int_{[0,\infty)}   \psi \left(\bP(-X > y \mid \sF_t)\right) \dif y+ \int_{(-\infty,0)}  \left[ \psi(\bP(-X > y \mid \sF_t)) -1 \right]  \dif y\\
		&=  \int_{[0,-a]}  \psi \left(\bP(-X > y \mid \sF_t)\right)\dif y+ \int_{[-b,0)}  \left[ \psi(\bP(-X > y \mid \sF_t)) -1 \right]  \dif y\\
		&\in [-b, -a], 
	\end{align*} 
	which means that $\rho_t^\omega$ is bounded and (ii) is proved. 
	
	Next we will show that $\rho^\psi$ satisfies all DCRM properties, that is, adapted, normalized, monotone, local, cash-additive, monotone decreasing, sub-additive and positive homogeneous. 
	
	\smallskip 
	\noindent
	\textit{Adaptiveness.} 
	For any $ \omega \in \Omega \backslash N^{X}$ and $y \in \mathbb{R}$,  $\bP(-X > y \mid \sF_t)(\omega)$ is a regular conditional distribution, and hence it is $\sF_t \bigotimes \mathcal{B}(\mathbb{R})$ joint measurable. Continuity of $\psi $ preserves joint measurability, by Tonelli theorem, $\omega \mapsto  \rho_t^\psi(X)(\omega)$ is $\sF_t$-measurable.
	
	\smallskip 
	\noindent
	\textit{Normalization.}
	\begin{align*}
		\rho_t^\psi(0)= \int_{[0,\infty)}   \psi \left(\bP(0 > y \mid \sF_t)\right)   \dif y+ \int_{(-\infty,0)}  \left[ \psi(\bP(0 > y \mid \sF_t)) -1 \right] \dif y= 0. 
	\end{align*} 
	
	\smallskip 
	\noindent
	\textit{Locality.}  		
	For any $ \omega \in \Omega \backslash  N^X $ and $A \in \sF_t$, note that, 
	\begin{align*}
		\bP(-\1_{A} X > y \mid \sF_t)(\omega) = 
		\begin{cases}
			\bP( A \cap \{ -X > y \} \mid \sF_t)(\omega), & y \in [0, \infty) \\
			\bP( \{ A \cap \{ -X > y \} \} \cup A^c \mid \sF_t)(\omega), & y \in (-\infty, 0). 
		\end{cases}
	\end{align*}
	Denote by $J_1= \1_{A}(\omega) \rho_t^\psi(\1_{A}  X)(\omega)$ and we will show that $J_1= \1_{A}(\omega) \rho_t^\psi(X)(\omega)$.
	\begin{align*}
		J_1  & = \1_{A}(\omega)  \int_{[0, \infty)}   \psi(\bP(-\1_{A} X > y \mid \sF_t)(\omega)) \dif y \\ 
		& \qquad +  \1_{A}(\omega)  \int_{(-\infty, 0)}   [ \psi(\bP(-\1_{A}  X > y\mid \sF_t)(\omega)) -1 ]  \dif y\\
		&=  \1_{A}(\omega)  \int_{[0, \infty)}   \psi(\bP( A \cap \{ -X > y \} \mid  \sF_t)(\omega))   \dif y \\
		& \qquad +  \1_{A}(\omega)  \int_{(-\infty, 0)}   [ \psi(\bP( \{ A \cap \{ -X > y \} \} \cup A^c \mid \sF_t)(\omega)) -1 ]  \dif y\\
		& = \1_{A}(\omega) \int_{[0, \infty)}   \psi(\bE( \1_A \1_{ -X > y } \mid \sF_t)(\omega))   \dif y \\
		& \qquad +    \1_{A}(\omega) \int_{(-\infty, 0)}   [ \psi(\bE( \1_A \1_{ -X > y } +\1_{A^c}  \mid \sF_t)(\omega)) -1 ]  \dif y.
	\end{align*} 
	Then, since $\psi(0)= 0$ and $\psi(1)=1$, we can move $\1_A$ outside of $\psi$ and continue the proof as following,  
	\begin{align*}
		J_1  & = \1_{A}(\omega)  \int_{[0, \infty)}   \1_A(\omega)  \psi(\bP( -X > y \mid\sF_t)(\omega))   \dif y \\
		& \qquad +   \1_{A}(\omega)  \int_{(-\infty, 0)}   [ \1_A(\omega) \psi(  \bP(  -X > y \mid \sF_t)(\omega))+ \1_{A^c}(\omega) -1 ] \dif y\\
		& =  \1_{A}(\omega)  \int_{[0, \infty)}  \psi(\bP( -X > y \mid \sF_t)(\omega))   \dif y \\
		& \qquad  +  \1_{A}(\omega) \int_{(-\infty, 0)}   [ \1_A(\omega) \psi(  \bP(  -X > y \mid \sF_t)(\omega))- \1_{A}(\omega) ]  \dif y\\
		& =  \1_{A}(\omega)  \int_{[0, \infty)}  \psi(\bP( -X > y \mid \sF_t)(\omega))   \dif y   \\
		& \qquad +  \1_{A}(\omega) \int_{(-\infty, 0)}   [ \psi(  \bP(  -X > y \mid \sF_t)(\omega))-1 ]  \dif y\\
		& =  \1_{A}(\omega) \rho_t^\psi(X)(\omega). 
	\end{align*}

	\smallskip\noindent
	\textit{Cash Additivity.} 		
	We first prove cash additivity for $m \in L^\infty_{t, +}$. For any~$m \in L^\infty_{t, +}$, there exists a sequence of step functions,
	$$
	m^n:= \sum_{i=1}^{N_n} y_i \1_{A_i}, \hspace{3mm} y_i \geq 0, \hspace{3mm} A_i \in \sF_t,
	$$
	increasingly converge to $m$ in $L^\infty$, where $N_n$ is a constant and $A_i \cap A_j= \emptyset$, for $i \neq j$. For any $X \in L^\infty$, with fixed $y \in \mathbb{R}$, $ \{ X+ m^n> y \}$ is an increasing sequence of sets converging to $ \{ X+ m> y \}$. Note that for any $n \in \bN$, there exists a null set $N^{X, m^n, t} \in \sF_t$, such that for any $\omega \in \Omega \backslash N^{X, m^n, t}$,  $\bP(X+ m^n> y \mid \sF_t)(\omega)$ is a regular distribution.
	
	Let $M= \bigcup_{n \in \bN} \bigcup_{t \in \cT} N^{X, m^n, t} \bigcup N^{X, m}$, which is a null set as a countable union of null sets.  For any $ \omega \in \Omega \backslash M$, by continuity of probability,
	\[
	\bP(X+ m^n> y \mid \sF_t)(\omega) \nearrow \bP(X+ m > y \mid \sF_t)(\omega), \quad n\to\infty. 
	\] 
	By left-continuity of $\psi$, as $n\to\infty$,
	\[
	\psi(\bP(X+ m^n> y \mid \sF_t)(\omega)) \nearrow \psi(\bP( X+ m> y \mid \sF_t)(\omega)).
	\]
	Then by monotone convergence theorem, as $n\to\infty$,
	\begin{equation} \label{eq:ca1}
		\int_{[0,\infty)} \psi (\bP(X+ m^n >y \mid \sF_t)(\omega)) \dif y \rightarrow \int_{[0,\infty)}  \psi (\bP(X+ m >y\mid\sF_t)(\omega)) \dif y. \\
	\end{equation}
	Note that for any $n \in \bN$ and $y \in \mathbb{R}_{-}$,
	\[
	\left|\psi(\bP(X+ m^n> y\mid\sF_t)(\omega))- 1 \right| \leq |\psi(\bP(X> y\mid\sF_t)(\omega))- 1|,
	\]
	and the latter is integrable on $(-\infty, 0)$. By dominated convergence theorem, as $n\to\infty$,
	\begin{equation} \label{eq:ca2}
		\int_{(-\infty,0)} [\psi(\bP(X+ m^n> y \mid \sF_t)(\omega))- 1] \dif y \rightarrow \int_{(-\infty,0)}   [\psi(\bP(X+ m> y\mid \sF_t)(\omega))- 1] \dif y.
	\end{equation}
	Combining \eqref{eq:ca1} and \eqref{eq:ca2}, and obtain
	\begin{equation} \label{eq:ca3}
		\rho_t^\psi(X+ m^n)(\omega) \rightarrow \rho_t^\psi(X+ m)(\omega), \quad n\to\infty.  
	\end{equation}
	It is enough to show that,
	\begin{equation} \label{eq:ca4}
		\rho_t^\psi(X+ m^n)(\omega) = \rho_t^\psi(X)(\omega)- m^n(\omega), 
	\end{equation}
	since this combined with \eqref{eq:ca3}  imply the desired equality $\rho_t^\psi(X+ m)(\omega)= \rho_t^\psi(X)(\omega)- m(\omega)$. Let us prove next \eqref{eq:ca4}, for which we put $J_2 := \rho_t^\psi(X+ m^n)$. Then, 
	\begin{align*}
		J_2 &= \int_{[0, \infty)}   \psi(\bP(-X- m^n > y\mid\sF_t))   \dif y+ \int_{(-\infty, 0)}   \left[ \psi(\bP(-X-m^n > y\mid\sF_t)) -1 \right]  \dif y\\
		&= \int_{[0, \infty)}   \psi\left(\bP\left(-X- \sum_{i=1}^{N_n} y_i \1_{A_i} > y\mid\sF_t\right)\right)   \dif y \\
		& \qquad \qquad + \int_{(-\infty, 0)}   \left[ \psi\left(\bP\left(-X- \sum_{i=1}^{N_n} y_i 1_{A_i} > y\mid\sF_t\right)\right) -1 \right]  \dif y\\
		&= \int_{[0, \infty)}   \psi\left(\bP\left(-X  > \sum_{i=1}^{N_n} (y+y_i)\1_{A_i} \mid\sF_t\right)\right)   \dif y \\
		&\qquad \qquad + \int_{(-\infty, 0)}   \left[ \psi\left(\bP\left(-X  > \sum_{i=1}^{N_n} (y+y_i)\1_{A_i} \mid\sF_t\right)\right) -1 \right]  \dif y\\
		&= \int_{[0, \infty)}   \psi\left(\bE\left(\1_{-X  > \sum_{i=1}^{N_n} (y+y_i)\1_{A_i}} \sum_{i=1}^{N_n}\1_{A_i} \mid\sF_t\right)\right)   \dif y \\
		& \qquad \qquad +  \int_{(-\infty, 0)}  \left[  \psi\left(\bE\left(\1_{-X  > \sum_{i=1}^{N_n} (y+y_i)\1_{A_i}} \sum_{i=1}^{N_n}\1_{A_i}\mid\sF_t\right)\right) -1 \right]  \dif y\\
		&= \int_{[0, \infty)}   \psi\left(\bE\left(\sum_{i=1}^{N_n} \1_{-X  > (y+y_i)} \1_{A_i} \mid\sF_t\right)\right)   \dif y \\
		& \qquad \qquad  + \int_{(-\infty, 0)}   \left[  \psi\left(\bE\left(\sum_{i=1}^{N_n} \1_{-X  > (y+y_i)} \1_{A_i}\mid\sF_t\right)\right) -1 \right]  \dif y\\
		&= \int_{[0, \infty)}   \psi\left(\sum_{i=1}^{N_n} \1_{A_i} \bE\left(\1_{-X  > (y+y_i)}\mid\sF_t\right)\right)  \dif y \\
		& \qquad \qquad + \int_{(-\infty, 0)}    \left[ \psi\left(\sum_{i=1}^{N_n} \1_{A_i} \bE\left(\1_{-X  > (y+y_i)}\mid\sF_t\right)\right) -1 \right]  \dif y\\
		&=  \sum_{i=1}^{N_n} \1_{A_i} \int_{[0, \infty)}  \psi\left( \bP\left(-X  > (y+y_i)\mid\sF_t\right)\right)   \dif y \\
		& \qquad \qquad +  \sum_{i=1}^{N_n} \1_{A_i} \int_{(-\infty, 0)}  \left[ \psi\left( \bP\left(-X  > (y+y_i)\mid\sF_t\right)\right) -1 \right]  \dif y.
	\end{align*}
	Take $z_i= y+y_i$, and since $y_i >0$,  we continue the proof as follows,
	\begin{align*}
		J_2 & =  \sum_{i=0}^{N_n} \1_{A_i} \left\{ \int_{[y_i, \infty)} \psi(\bP(-X > z_i\mid\sF_t))   \dif z_i+ \int_{(-\infty, y_i)}   [ \psi(\bP(-X > z_i\mid\sF_t)) -1 ]  \dif z_i \right\} \\
		& = \sum_{i=0}^{N_n} \1_{A_i} \left\{ \int_{[y_i, \infty)} \psi(\bP(-X > z_i\mid\sF_t))   \dif z_i+ \int_{[0, y_i)} \left[ \psi(\bP(-X > z_i\mid\sF_t)) -1 \right]  \dif z_i \right\}\\
		& \qquad \qquad + \sum_{i=0}^{N_n} \1_{A_i} \left\{ \int_{(-\infty, 0)}   [ \psi(\bP(-X > z_i\mid\sF_t)) -1 ]  \dif z_i \right\} \\
		& = \sum_{i=0}^{N_n} \1_{A_i} \left\{ \int_{[0, \infty)} \psi(\bP(-X > z_i\mid\sF_t))   \dif z_i+ \int_{(-\infty, 0)}   \left[ \psi(\bP(-X > z_i\mid\sF_t)) -1 \right]  \dif z_i  \right\} \\
		& \qquad \qquad -\sum_{i=0}^{N_n} y_i \1_{A_i} \\
		&= \int_{[0, \infty)}   \psi(\bP(-X > z\mid\sF_t))   \dif z+ \int_{(-\infty, 0)}   \left[ \psi(\bP(-X > z\mid\sF_t)) -1 \right]  \dif z- m^n\\
		&= \rho_t^\psi(X)- m^n.
	\end{align*}
	Thus, \eqref{eq:ca4} is proved.
	
	For general $m\in L^\infty$, denote by $M$ its essential supremum, and let let $\hat{m}= m+  M \geq 0 $. By cash additivity of $\rho_t^\psi$ for non-negative $\hat{m}$,
	\begin{equation} \label{eq:ca5}
		\rho_t^\psi(X+ \hat{m})(\omega)= \rho_t^\psi(X)(\omega)- \hat{m}(\omega)= \rho_t^\psi(X)(\omega)- m(\omega)- M.
	\end{equation}
	Also by cash additivity of $\rho_t^\psi$ for constant $M$, i.e. $\rho_t^\psi(X+ M)(\omega)= \rho_t^\psi(X)(\omega)- M$, we have
	\begin{equation} \label{eq:ca6}
		\rho_t^\psi(X+ \hat{m})(\omega)= \rho_t^\psi(X+ m+ M)(\omega)= \rho_t^\psi(X+ m)(\omega)- M.
	\end{equation}
	Combine \eqref{eq:ca5} and \eqref{eq:ca6}, and deduce that $\rho_t^\psi(X+ m)(\omega)= \rho_t^\psi(X)(\omega)- m(\omega)$. Thus, cash-additivity is established.
	
	\smallskip\noindent
	\textit{Monotonicity.} For any $X, Y \in L^\infty$ with $X \leq Y $, we have for any $ y \in \mathbb{R}$, $\bP(-X > y\mid\sF_t) \geq \bP(-Y > y\mid\sF_t)$. Since $\psi$ is increasing,  $\psi(\bP(-X > y\mid\sF_t)) \geq \psi(\bP(-Y > y\mid\sF_t))$. Then we have,
	\begin{align*}
		\rho_t^\psi(X)&= \int_{[0, \infty)}   \psi(\bP(-X > y\mid\sF_t))   \dif y+ \int_{(-\infty, 0)}   [ \psi(\bP(-X > y\mid\sF_t)) -1 ]  \dif y \\
		&\geq \int_{[0, \infty)}   \psi(\bP(-Y > y\mid\sF_t))   \dif y+ \int_{(-\infty, 0)}  [ \psi(\bP(-Y > y\mid\sF_t)) -1 ]  \dif y = \rho_t^\psi(Y).
	\end{align*}
	
	\smallskip\noindent
	\textit{Sub-additivity.} This part of the proof is inspired by \cite[Lemma 4.92]{FollmerSchiedBook2004}.

	First,  we observe that $\rho_t^\psi$ is Lipschitz continuous with respect to the $L^\infty$ norm, i.e. for  given $X$, $Y \in L^\infty$,
	\begin{equation}\label{eq:Lip}
		\norm{\rho_t^\psi(X) - \rho_t^\psi(Y)}_\infty  \leq \norm{X- Y}_\infty.
	\end{equation}
	Indeed, since $X \leq Y+ \norm{X- Y }_\infty$, by monotonicity of $\rho_t^\psi$, we have
	$\rho_t^\psi(X) \geq \rho_t^\psi(Y + \norm{X- Y}_\infty)$, and according to cash additivity of $\rho_t^\psi$, $	\rho_t^\psi(X) \geq \rho_t^\psi(Y)- \norm{X- Y}_\infty.$ Reversing the roles of $X$ and $Y$ yields   \eqref{eq:Lip}.

	In view of Lemma~\ref{le:ineqDCRM},  given $\psi \in \Upsilon$ and $t \in \mathcal{T}$, for any $\omega \in \Omega \backslash N^X$, we build the probability measure $Q^{\omega}$, such that for any non-negative step function $X = \sum_{i=1}^{n} x_{i} \1_{A_{i}}$,
	\[ 
	\sum_{i=1}^{n} -x_{i}Q^{\omega}[A_{i}] \leq  \rho_t^\psi(X)(\omega).
	\]
	In particular, the equality holds if $x_{1} \leq \ldots \leq x_{n}$. Now if $X$ is not non-negative, take $\hat{X} = X +  M \geq 0 $,  where $M= \esssup X$. By cash additivity of $\rho_t^\psi$,
	\[
	\rho_t^\psi(\hat{X})(\omega)= \rho_t^\psi(X + M)(\omega)= \rho_t^\psi(X)(\omega)- M.
	\]
	Since $\hat{X} \geq 0$,
	\[
	\rho_t^\psi(\hat{X})(\omega) \geq \sum_{i=1}^{n} [-x_{i}- N] Q^{\omega}(A_{i})= \sum_{i=1}^{n} -x_{i} Q^{\omega}(A_{i}) - M.
	\]
	Then we have,
	\[
	\rho_t^\psi(X)(\omega) \geq \sum_{i=1}^{n} -x_{i} Q^{\omega}(A_{i}).
	\]
	and therefore Lemma~\ref{le:ineqDCRM} holds true all step functions.
	
	Now back to the proof of sub additivity of $\rho_t^\psi$. Given $X$, $Y \in L^\infty$, there exist two sequences of step functions $X^n$ and $Y^n$ converging to $X$ and $Y$ respectively in $L^\infty$. Then, by \eqref{eq:Lip}, to prove sub-additivity, it is enough to show that,  	for any $\omega \in \Omega \backslash N^X$,
	\begin{equation} \label{eq:subine1}
		\rho_t^\psi(X^n + Y^n)(\omega) \leq \rho_t^\psi(X^n)(\omega) + \rho_t^\psi(Y^n)(\omega).
	\end{equation}
	For fixed $n$, let $A_1, \ldots , A_{m}$ be the (disjoint) partition of underlying probability space $\Omega$, such that $X^n= \sum_{i=1}^{m} x_i \1_{A_i}$ and $Y^n = \sum_{i=1}^{m} y_i \1_{A_i}$,  and  assume that the indices $i = 1,\ldots, m$ are arranged such that $x_1 + y_1 \leq \ldots \leq x_m + y_m $. In view of  generalized version of Lemma~\ref{le:ineqDCRM}, for any $\omega \in \Omega \backslash N^X$, one can construct a probability $Q^{\omega}$, such that,
	\[
	\rho_t^\psi(X^n + Y^n)(\omega) = \sum_{i=1}^{m} (-x_i-y_i) Q^{\omega}[A_{i}] = \sum_{i=1}^{m} -x_iQ^{\omega}[A_{i}] -y_iQ^{\omega}[A_{i}].
	\]
	and
	\begin{align*}
		\sum_{i=1}^{m}  -x_iQ^{\omega}[A_{i}] \leq \rho_t^\psi(X^n)(\omega), \quad 
		\sum_{i=1}^{m}  -y_iQ^{\omega}[A_{i}] \leq \rho_t^\psi(Y^n)(\omega). 		
	\end{align*}
	In conclusion,
	\[
	\rho_t^\psi(X^n + Y^n)(\omega) = \sum_{i=1}^{m} -x_iQ^{\omega}[A_{i}] -y_iQ^{\omega}[A_{i}]  \leq \rho_t^\psi(X^n)(\omega) + \rho_t^\psi(Y^n)(\omega),
	\]
	and thus \eqref{eq:subine1} is proved.

	\smallskip\noindent
	\textit{ Positive Homogeneity.}
	For any $\lambda \in L^\infty_{t, +}$, there exists a sequence of step functions,
	\[
	\lambda^n:= \sum_{i=1}^{N_n} y_i \1_{A_i}, \hspace{3mm} y_i \geq 0, \hspace{3mm} A_i \in \sF_t,
	\]
	increasing and convergent to $\lambda$ in $L^\infty$, where $N_n$ is a constant and $A_i$ are disjoint. Note that for any $n \in \bN$, there exists a null set $N^{X, \lambda^n, t} \in \sF_t$, such that for any $\omega \in \Omega \backslash N^{X, \lambda^n, t}$,  $\bP(\lambda^n X> y \mid \sF_t)(\omega)$ is a regular distribution.
	
	Let $M= \bigcup_{n \in \bN} \bigcup_{t \in \cT} N^{X, \lambda^n, t} \bigcup N^{X, \lambda}$, since $M$ is a countable union of null sets, it is still a null set. For any $X \in L^\infty$ and fixed $y \in [0, \infty)$, $ \{ \lambda^n X> y \}$ is an increasing sequence of sets. For any  $ \omega \in \Omega \backslash M$, by continuity of probability for increasing sequence of sets,
	\[
	\bP(\lambda^n X > y\mid\sF_t)(\omega) \nearrow \bP(\lambda X > y\mid\sF_t)(\omega), \quad n\to\infty.  
	\] 
	By continuity of $\psi$, as $n\to\infty$,
	\[
	\psi(\bP(\lambda^n X> y\mid\sF_t)(\omega)) \nearrow \psi(\bP( \lambda X> y\mid\sF_t)(\omega)).
	\]
	In view of the monotone convergence theorem,
	\begin{equation} \label{eq:ph1}
		\int_0^\infty \psi (\bP(\lambda^n X >y\mid\sF_t)(\omega)) \dif y \xrightarrow[n\to\infty]{} \int_0^\infty  \psi (\bP(\lambda X >y\mid\sF_t)(\omega)) \dif y.
	\end{equation}
	Similarly, one can show that 
	\begin{equation} \label{eq:ph2}
		\int_{-\infty}^0  [\psi(\bP(\lambda^n X> y\mid\sF_t)(\omega))- 1] \dif y \xrightarrow[n\to\infty]{} \int_{-\infty}^0   [\psi(\bP(\lambda X> y\mid\sF_t)(\omega))- 1] \dif y.
	\end{equation}
	Combining \eqref{eq:ph1} and \eqref{eq:ph2},  we obtain,
	\[
	\rho_t^\psi(\lambda^n X)(\omega) \xrightarrow[n\to\infty]{} \rho_t^\psi(\lambda X)(\omega).
	\]
	From here, the desired equality $\rho_t^\psi(\lambda X)(\omega)= \lambda \rho_t^\psi(X)(\omega)$ will follow once we show that 
	\[
	\rho_t^\psi(\lambda^n X)(\omega) = \lambda^n \rho_t^\psi(X)(\omega),
	\]
	which we prove next. 
	\begin{align*}
		\rho_t^\psi(\lambda^n X)&= \int_{[0, \infty)}   \psi(\bP(-\lambda^n X > y\mid\sF_t))   \dif y+ \int_{(-\infty, 0)}   [ \psi\left(\bP\left(-\lambda^n X > y\mid\sF_t\right)\right) -1 ]  \dif y\\
		&= \int_{[0, \infty)}   \psi\left(\bP\left(-\sum_{i=1}^{N_n} y_i \1_{A_i} X > y\mid\sF_t\right)\right)   \dif y \\
		& \qquad \qquad  + \int_{(-\infty, 0)}   \left[ \psi\left(\bP\left(-\sum_{i=1}^{N_n} y_i \1_{A_i} X > y\mid\sF_t\right)\right) -1 \right]  \dif y\\
		&= \int_{[0, \infty)}   \psi\left(\bP\left(-X  > \sum_{i=1}^{N_n} (y/y_i)\1_{A_i} \mid\sF_t\right)\right)   \dif y \\
		& \qquad \qquad + \int_{(-\infty, 0)}   \left[ \psi\left(\bP\left(-X  > \sum_{i=1}^{N_n} (y/y_i)\1_{A_i} \mid\sF_t\right)\right) -1 \right]  \dif y\\
		&= \int_0^{\infty}   \psi\left(\bE\left(\1_{-X  > \sum_{i=1}^{N_n} (y/y_i)\1_{A_i}} \sum_{i=1}^{N_n}\1_{A_i} \mid\sF_t\right)\right)   \dif y\\
		& \qquad \qquad + \int_{-\infty}^0   \left[ \psi\left(\bE\left(\1_{-X  > \sum_{i=1}^{N_n} (y/y_i)\1_{A_i}} \sum_{i=1}^{N_n}\1_{A_i}\mid\sF_t\right)\right) -1 \right]  \dif y\\
		&= \int_0^{\infty}   \psi\left(\bE\left(\sum_{i=1}^{N_n} \1_{-X  > y/y_i} \1_{A_i}\mid\sF_t\right)\right)   \dif y \\
		& \qquad \qquad  + \int_{-\infty}^0   \left[ \psi\left(\bE\left(\sum_{i=1}^{N_n} \1_{-X  > y/y_i} \1_{A_i}\mid\sF_t\right)\right) -1 \right]  \dif y\\
		&= \int_0^{\infty}   \psi\left(\sum_{i=1}^{N_n} \1_{A_i} \bE\left(\1_{-X  > y/y_i}\mid\sF_t\right)\right)   \dif y \\
		& \qquad \qquad + \int_{-\infty}^0   \left[ \psi\left(\sum_{i=1}^{N_n} \1_{A_i} \bE\left(\1_{-X  > y/y_i}\mid\sF_t\right)\right) -1 \right]  \dif y\\
		&=  \sum_{i=1}^{N_n} \1_{A_i} \int_0^{\infty}  \psi\left(\bP\left(-X  > y/y_i\mid\sF_t\right)\right)   \dif y \\
		& \qquad \qquad +  \sum_{i=1}^{N_n} \1_{A_i} \int_{-\infty}^0   [ \psi\left(\bP\left(-X  > y/y_i)\mid\sF_t\right)\right) -1 ]  \dif y\\
		&\overset{z_i= y/y_i}{=} \sum_{i=1}^{N_n} \1_{A_i} y_i \Big\{ \int_{0}^{\infty} \psi\left(\bP\left(-X > z_i\mid\sF_t\right)\right)   \dif z_i\\
		& \qquad \qquad + \int_{-\infty}^{0}   [ \psi\left(\bP\left(-X > z_i\mid\sF_t\right)\right) -1 ]  \dif z_i \Big\} \\
		&= \lambda^n \rho_t^\psi(X).
	\end{align*}
	Note that in the third equality, we assume $y_i > 0$ for all $i\in I$, otherwise just omit the terms in the sums with $y_j=0$.
	
	The proof is complete.
\end{proof}

Next results gives a convenient representation of the DCRMs generated by distortion functions.

\begin{lemma}
	For any $\psi\in\Upsilon$, the following representation holds 
	\begin{equation} \label{eq:ChoquetDCRM-supp2}
		\rho_t^\psi(X)= - \int_{\mathbb{R}} y \dif \psi(\bP(X \leq y\mid\sF_t)). 
	\end{equation}
\end{lemma} 
\begin{proof}
	In the space $BV$ of functions with bounded variation, consider the set 
	\[
	NBV:= \{ F \in BV \mid F \hspace{2mm} \text{is right continuous and} \hspace{2mm} F(-\infty)= 0 \}.  
	\] 
	It can be shown (cf. \cite[Theorem 3.36, and Exercise 34(b), p.108]{Folland1999}) that if $F, G\in NBV$ and there are no points in $[a, b]$ where $F$ and $G$ are both discontinuous, then, 
	\[
	\int_{[a, b]} F(y)  \dif G(y)= F(b)G(b)- F(a-)G(a-)- \int_{[a, b]} G(y) \dif F(y).
	\]
	Let $a$ and $b$  be the essential infimum and the essential supremum of $X$ respectively and assume that $a < 0 <  b$. Notice that (potential) discontinuities in the distribution function $F_X$ are at most countable, and hence the Lebesgue integrals below are not affected. For any $\omega \in \Omega \backslash N^{X}$, we deduce
	\begin{align*}
		\rho_t^\psi(X)(\omega) & = \int_{[0,-a]} \psi(\bP(-X > y\mid\sF_t)(\omega)) \dif y + \int_{[-b, 0)} (\psi(\bP(-X > y\mid\sF_t)(\omega))- 1) \dif y \\ 
		& = \int_{[0,-a]} \psi(\bP(X < -y\mid\sF_t)(\omega)) \dif y + \int_{[-b, 0)} \left(\psi(\bP(X < -y\mid\sF_t)(\omega))- 1 \right) \dif y  \\
		& = \int_{[0,-a]} \psi(\bP(X \leq -y\mid\sF_t)(\omega)) \dif y + \int_{[-b, 0)} \left(\psi(\bP(X \leq -y\mid\sF_t)(\omega))- 1 \right) \dif y  \\
		& = \int_{[a,0]} \psi(\bP(X \leq z\mid\sF_t)(\omega)) \dif z + \int_{(0, b]} \left(\psi(\bP(X \leq z\mid\sF_t)(\omega)) \right) \dif z -b \\
		& = \psi(\bP(X \leq z\mid\sF_t)(\omega)) z|_{a-}^0- \int_{[a,0]} z \dif\psi(\bP(X \leq z\mid\sF_t)(\omega)) \\ 
		& \qquad \qquad +  \psi(\bP(X \leq z\mid\sF_t)(\omega))  z|_{0}^b- \int_{(0, b]} z \dif\left(\psi(\bP(X \leq z\mid\sF_t)(\omega)) \right)-b\\
		& = - \int_{[a,0]} z \dif\psi(\bP(X \leq z\mid\sF_t)(\omega)) - \int_{(0, b]} \dif \left(\psi(\bP(X \leq z\mid\sF_t)(\omega)) \right)  \\
		& =  - \int_{[a, b]}  z   \dif\psi(\bP(X \leq z\mid\sF_t)(\omega)). 
	\end{align*} 
\end{proof}

\begin{remark}
	We note that left continuity of $\psi$ is sufficient to prove positive homogeneity of $\rho^{\psi}_t$, however the right continuity of $\psi$ is  not only needed for proving cash additivity and positive homogeneity of $\rho^{\psi}_t$, but also serves as a sufficient condition for representation \eqref{eq:ChoquetDCRM-2} to hold.
\end{remark}

\subsection{On DCAIs}

In this section we present some fundamental properties on DCAIs. Throughout this section  we assume additionally the probability space is atomless. One of the reasons to invoke atomless property is due to the fact that a probability space supports random variables with continuous distribution if and only if such probability space is atomless; cf. \cite[Proposition~A.27]{FollmerSchiedBook2004}. 

We start with several results about families of distortion functions and the corresponding DCRMs.

\begin{lemma}   \label{le:cexin-supp}
	A family of DCRMs $(\rho^{\psi_x})_{x >0}$ is increasing if and only if $(\psi_{x})_{x >0}$ is an increasing family of distortion functions. 
	Moreover, 	a family of DCRMs $(\rho^{\psi_x})_{x >0}$ is right continuous if and only if  $(\psi_{x})_{x >0}$ is a right continuous family of distortion functions. 
\end{lemma}
\begin{proof} 
	We start with the first part of the assertion on monotonicity. 
	
	\noindent ($\Leftarrow$) Using Choquet representation \eqref{eq:ChoquetDCRM} of $\rho_t^{\psi_x}$, for any $X \in L^\infty$, $\omega \in \Omega \backslash N^X$, $t \in \mathcal{T}$ and $0< x_1 \leq x_2$, we get 
	\begin{align*}
		\rho^{\psi_{x_2}}_t(X)(\omega)- &\rho^{\psi_{x_1}}_t(X)(\omega)=
		\int_{[0,\infty)}  \left[ \psi_{x_2}(\bP(-X > y\mid\sF_t)(\omega))- \psi_{x_1}(\bP(-X > y\mid\sF_t)(\omega)) \right]   \dif y\\
		&+ \int_{(-\infty,0)}  \left[ \psi_{x_2}(\bP(-X > y\mid\sF_t)(\omega))- \psi_{x_1}(\bP(-X > y\mid\sF_t)(\omega)) \right]  \dif y \geq 0.
	\end{align*} 
	($\Rightarrow$) Proof by contradiction. Assume that there exist   $0< x_1 \leq x_2$ and $y^{\star} \in [0, 1]$, such that $\psi_{x_1}(y^{\star} ) > \psi_{x_2}(y^{\star} )$. Since the probability space is atomless, we take $X= 1$ with probability $y^{\star}$ and $X= 2$ with probability $1-y^{\star}$. When $t= 0$,  
	\begin{align*}
		\rho^{\psi_{x_2}}_0(X)- \rho^{\psi_{x_1}}_0(X) & = \int_{[0,\infty)}  \left[ \psi_{x_2}(\bP(-X > y))- \psi_{x_1}(\bP(-X > y)) \right]   \dif y\\
		& \qquad + \int_{(-\infty,0)}  \left[ \psi_{x_2}(\bP(-X > y))- \psi_{x_1}(\bP(-X > y)) \right]  \dif y \\
		&=  \int_{[-2,-1)}  \left[ \psi_{x_2}(\bP(-X > y))- \psi_{x_1}(\bP(-X > y)) \right]  \dif y \\
		&=    \psi_{x_2}(y^{\star})- \psi_{x_1}(y^{\star})  < 0,
	\end{align*} 
	which contradicts the fact that $(\rho^{\psi_x})_{x >0}$ is an increasing  family of DCRMs.

	\smallskip\noindent
	Next we focus on the second part of the lemma about the right continuity. 
	
	\noindent ($\Leftarrow$)	Let $a$ and $b$  be the essential infimum and the essential supremum of $X$ respectively and assume that $a < b$.  For any  $z> x$ and $ \omega \in \Omega \backslash N^{X}$, 
	\begin{align*} 
		\rho^{\psi_{z}}_t(X)(\omega)&= \int_{[0,\infty)} \psi_{z}(\bP(-X > y\mid\sF_t)(\omega))   \dif y+ \int_{(-\infty,0)}  \left[ \psi_{z}(\bP(-X > y\mid\sF_t)(\omega))- 1 \right] \dif y \\
		&= \int_{[0,-a)} \psi_{z}(\bP(-X > y\mid\sF_t)(\omega))   \dif y+ \int_{(-b,0)}   \left[ \psi_{z}(\bP(-X > y\mid\sF_t)(\omega)) - 1 \right] \dif y. 
	\end{align*} 
	For any $y \in [0,-a)$, $ |\psi_z(\bP(-X> y|\sF_t)(\omega))| \leq 1$, and thus 
	\[
	\int_{[0,-a)} |\psi_{z}(\bP(-X > y|\sF_t)(\omega))|   \dif y \leq \int_{[0,-a)} 1  \dif y= a. 
	\] 
	Similarly, for any $y \in (-b,0)$, $ |\psi_z(\bP(-X> y|\sF_t)(\omega))- 1| \leq 1$, then  
	\[
	\int_{(-b,0)} |\psi_z(\bP(-X> y|\sF_t)(\omega))- 1|   \dif y \leq b. 
	\]
	Since $(\psi_{x})_{x >0}$ is a right continuous family of distortion functions, then by the dominated convergence theorem, 
	\begin{gather*}	
		\lim_{z \rightarrow x^{+}} \int_{[0,-a)} \psi_{z}(\bP(-X > y|\sF_t)(\omega))  \dif y =\int_{[0,-a)} \psi_{x}(\bP(-X > y|\sF_t)(\omega))   \dif y, \\
		\lim_{z \rightarrow x^{+}} \int_{(-b,0)} [\psi_{z}(\bP(-X> y|\sF_t)(\omega))- 1] \dif y = \int_{(-b,0)}   [\psi_{x}(\bP(-X> y|\sF_t)(\omega))- 1] \dif y. 
	\end{gather*}
	Combining the  above two identities, we have, 
	\[
	\lim_{z \rightarrow x^{+}}  \rho^{\psi_{z}}_t(X)(\omega)=  \rho^{\psi_{x}}_t(X)(\omega).
	\]
	
	\noindent ($\Rightarrow$) Proof by contradiction. Assume that there exists   $x_0 >0$ and $y^{\star} \in [0, 1]$, such that $\lim\limits_{{z \rightarrow x_0^{+}}} \psi_{z}(y^{\star}) \neq \psi_{x_0}(y^{\star})$. Since the probability space is atomless, we take $X= 1$ with probability $y^{\star}$ and $X= 2$ with probability $1-y^{\star}$. When $t= 0$,  
	\begin{align*} 
		\lim_{z \rightarrow {x_0}^{+}} \rho^{\psi_{z}}_0(X) & = \lim_{z \rightarrow x_0^{+}} \left[ \int_{[0,\infty)}   \psi_{z}(\bP(-X > y))   + \int_{(-\infty,0)}  \left[ \psi_{z}(\bP(-X > y))- 1 \right]  \dif y \right] \\
		& =  \lim_{z \rightarrow x_0^{+}} \int_{[-2,-1)}  \left[ \psi_{z}(\bP(-X > y)) -1 \right]  \dif  y \\
		& =    \lim_{z \rightarrow x_0^{+}} \psi_{z}(y^{\star})- 1  \\
		& \neq \psi_{x_0}(y^{\star})- 1 \\
		& =  \rho^{\psi_{x_0}}_0(X),
	\end{align*} 
	which contradicts the fact that $(\rho^{\psi_x})_{x >0}$ is a right continuous  family of DCRMs. 
	
	This completes the proof. 
\end{proof}

\begin{remark} 
	The assumption at the beginning of this section that the probability space is atomless  is important in Lemma~\ref{le:cexin} to ensure that $X= 1$ with probability $y^{\star}$ exists. If atomless property does not hold, then generally speaking the `only if' directions in Lemma~\ref{le:cexin} does not hold true, as Example~\ref{ex:noincreasing} and Example~\ref{ex:norightcon} below show. However, it is worth mentioning that a probability space with atoms may still support  Lemma~\ref{le:cexin} for random variables with distributions not affected by atoms. 
\end{remark}

\begin{example} \label{ex:noincreasing} 
	We will show that, generally speaking, an increasing family of DCRMs $(\rho^{\psi_x})_{x >0}$ does not imply that the corresponding family of distortions $(\psi_{x})_{x >0}$ is increasing.
	
	Let $\mathcal{T}:= \{0, 1\}$, and consider the filtered probability space $(\Omega, \sF, \{\sF_t\}_{t \in \mathcal{T}} , \bP)$, with $\Omega= \{ \omega_1, \omega_2\}$,  $\sF_0 = \{ \emptyset, \Omega \}$, $\sF_1 = \sF= 2^{\Omega}$ and $\mathbb{P}(\omega_1)= \mathbb{P}(\omega_2)= \frac{1}{2}$. 
	
	Suppose $m>0>n$, and let $X= m$ with probability $\frac{1}{2}$ and $X= n$ with probability $\frac{1}{2}$. For any $0< x_1 \leq x_2$ and $X \in L^\infty$,  $\rho^{\psi_{x_2}}_1(X)- \rho^{\psi_{x_1}}_1(X)= -X+ X= 0$. On the other hand, 
	\begin{align*}
		\rho^{\psi_{x_2}}_0(X)- \rho^{\psi_{x_1}}_0(X) &= \int_{[0,\infty)}  \left[ \psi_{x_2}(\bP(-X > y))- \psi_{x_1}(\bP(-X > y)) \right]   \dif y\\
		& \qquad + \int_{(-\infty,0)}  \left[ \psi_{x_2}(\bP(-X > y))- \psi_{x_1}(\bP(-X > y)) \right]  \dif y \\
		& = \int_{[0,-n)}  \left[ \psi_{x_2}(\bP(-X > y))- \psi_{x_1}(\bP(-X > y)) \right]  \dif y\\
		&\qquad + \int_{(-m,0)}  \left[ \psi_{x_2}(\bP(-X > y))- \psi_{x_1}(\bP(-X > y)) \right]  \dif y \\
		& =   (m-n) \left( \psi_{x_2}\left(\frac{1}{2}\right)- \psi_{x_1}\left(\frac{1}{2}\right) \right). 
	\end{align*} 
	Thus, $(\rho^{\psi_x})_{x >0}$ is an increasing  family of DCRMs, as long as  $ \psi_{x_1}(\frac{1}{2}) \leq \psi_{x_2}(\frac{1}{2})$, for any  $0< x_1 \leq x_2$,  
	Clearly, one can take $[y_1, y_2] \subset [0, 1]$ with $y_1> \frac{1}{2}$, such that $\psi_{x_1}(y) > \psi_{x_2}(y)$ for all $y \in (y_1, y_2)$ and $\psi_{x_1}(y) \leq \psi_{x_2}(y)$ on remaining interval, which makes $(\psi_{x})_{x >0}$ not an increasing family of distortion functions. 	
\end{example}

\begin{example} \label{ex:norightcon}
	Consider the same probability space as in Example~\ref{ex:noincreasing}. 
	We will show that a right continuous family of DCRMs $(\rho^{\psi_x})_{x >0}$ does not imply that $(\psi_{x})_{x >0}$ is a right continuous family of distortion functions.
	
	Suppose $m>0>n$, and let $X= m$ with probability $\frac{1}{2}$, and $X= n$ with probability $\frac{1}{2}$. For any $x> 0$ and $X \in L^\infty$,   $ \lim\limits_{z \rightarrow x^{+}} \rho^{\psi_{z}}_1(X)- \rho^{\psi_{x}}_1(X)= -X+ X= 0$, while 
	\begin{align*} 
		\lim_{z \rightarrow {x}^{+}} \rho^{\psi_{z}}_0(X) & = \lim_{z \rightarrow x^{+}} \left[ \int_{[0,\infty)}   \psi_{z}(\bP(-X > y))   + \int_{(-\infty,0)}  \left[ \psi_{z}(\bP(-X > y))- 1 \right]  dy \right] \\
		& =  \lim_{z \rightarrow x^{+}} \left[ \int_{[0,-n)}   \psi_{z}(\bP(-X > y)) dy+ \int_{(-m,0)}  \left[ \psi_{z}(\bP(-X > y))- 1\right]  dy  \right] \\
		& =  (m-n) \lim_{z \rightarrow x^{+}} \psi_{z}\left(\frac{1}{2}\right)- m. 
	\end{align*} 
	Hence, as long as  for any $x> 0$, $ \lim\limits_{z \rightarrow x^{+}} \psi_{z}(\frac{1}{2})= \psi_{x}(\frac{1}{2})$, we have $(\rho^{\psi_x})_{x >0}$ is a right continuous  family of DCRMs. Obviously, by taking $\psi_x, x>0$, such that $\lim\limits_{z \rightarrow x^{+}} \psi_{z}(y)\neq \psi_{x}(y)$ for any $y \neq \frac{1}{2}$,  implies that $(\psi_{x})_{x >0}$ is not a right continuous family of distortion functions. 
\end{example}

We are ready to present the main result of this section. 

\begin{proposition}
	The mapping  $\alpha_{t}^{\Psi}$ is a DCAI, that is $\alpha_{t}^{\Psi}$ is  adapted, local, quasi-concave, monotone increasing and scale invariant.
\end{proposition}

\begin{proof} The proof of this result can be obtained by using robust representations results for DCAIs (cf. \cite{BiaginiBion-Nadal2012} or \cite{BCDK2013}). For the sake of completeness, we present here an independent proof by verifying directly the corresponding properties of DCAI.

	\smallskip
	\noindent \textit{Adaptiveness.}
	Consider the set $A^x:= \{ \omega \in \Omega \mid \alpha_{t}^{\Psi}(X)(\omega) \geq x \}$, where $x \in \mathbb{R}_+$, $t \in \mathcal{T}$ and $X \in L^\infty$ are fixed.  We will show $A^x \in \sF_t $.
	Let
	\begin{align}
		B^x =  \{ \omega \in \Omega \mid \rho_t^{\psi_x}(X)(\omega) \leq 0 \}.
	\end{align}	
	First we will prove that $A^x=B^x$. Indeed, for any $\omega \in \{ \rho_t^{\psi_x}(X) \leq 0 \}$, we have $ x \in \left\{z \in \mathbb{R}_+ \mid \rho_t^{\psi_z}(X)(\omega) \leq 0 \right \}$. Thus,
	\begin{align} \label{eq:alphainrho}
		\alpha_{t}^{\Psi}(X)(\omega) = \sup \left\{z \in \mathbb{R}_+ \mid \rho_t^{\psi_z}(X)(\omega) \leq 0 \right \} \geq x,
	\end{align}
	and hence $B^x\subset A^x$.

	Next we show that $A^x \subset B^x$. By \eqref{eq:alphainrho},
	$$
	\omega \in A^x  \Rightarrow \sup \left\{z \in \mathbb{R}_+ \mid \rho_t^{\psi_z}(X)(\omega) \leq 0 \right \} \geq x.
	$$
	For any $y < x$, we also have $\sup \left\{z \in \mathbb{R}_+ \mid \rho_t^{\psi_z}(X)(\omega) \leq 0 \right \} > y$, which implies that  $\rho_t^{\psi_y}(X)(\omega) \leq 0 $. 
	Therefore, to show $ \omega \in  B^x$,   since $\rho_t^{\psi_x}(X)(\omega) \leq 0$, it suffices to show that 
	\begin{align} \label{eq:3.4}
		\lim_{y \uparrow x} \rho_t^{\psi_y}(X)(\omega)=  \rho_t^{\psi_x}(X)(\omega).
	\end{align}	
	To prove the latter, recall that $\left(\psi_{x}\right)_{x \in \mathbb{R}_{+}} $ is a family of continuous distortions pointwise increasing
	in  $x$, and thus by the monotone convergence theorem, we obtain,
	\begin{align} \label{eq:3.5}
		\lim_{y \uparrow x} \int_0^{\infty}   \psi_y(\bP(-X > z \mid \sF_t)(\omega))   \dif  z= \int_0^{\infty}   \psi_x(\bP(-X > z \mid \sF_t)(\omega)) \dif z.
	\end{align}
	Then since,
	\[
	\lim_{y \uparrow x} \left( \psi_y(\bP(-X > z\mid \sF_t)(\omega))- 1\right)= \psi_x(\bP(-X > z\mid \sF_t)(\omega))- 1,
	\]
	and for any $z \in (-\infty, 0)$ and $y \in \mathbb{R}_+$, we have that 
	$| \psi_y(\bP(-X > z\mid \sF_t)(\omega))- 1 |  \leq 1.$
	Again, by the dominated convergence theorem,
	\begin{align} \label{eq:3.6}
		\hspace{-1.1cm}	\lim_{y \uparrow x} \int^0_{-\infty}   [\psi_y(\bP(-X > z\mid \sF_t)(\omega))- 1]  \dif z= \int^0_{-\infty}   [\psi_x(\bP(-X > z\mid \sF_t)(\omega))- 1] \dif z.
	\end{align}
	With \eqref{eq:3.5} and \eqref{eq:3.6} at hand,  \eqref{eq:3.4} follows at once, and thus 
	\begin{equation}\label{eq:3-11}
		A^x = B^x
	\end{equation}
	is proved.
	Measurability of $\rho_t^{\psi_x}(X)$ is verified in the proof of Proposition~\ref{prop:DCRM}, which implies that $A^x \in \sF_t$.

	\smallskip
	\noindent\textit{Locality.}
	For any  $t \in \mathcal{T}$, $X \in L^\infty$, $A \in \sF_t$ and $ \omega \in \Omega \backslash N^X$, 	
	\begin{align*}
		1_A(\omega) \alpha_t^{\Psi}(X)(\omega) &= 1_A(\omega) \sup \Big\{x \in \mathbb{R}_+ \mid \int_0^{\infty}   \psi_x(\bP(-X > y\mid \sF_t)(\omega))   \dif y\\
		& \qquad + \int^0_{-\infty}  [ \psi_x(\bP(-X > y\mid \sF_t)(\omega)) -1 ]  \dif y \leq 0  \Big\} \\
		&= 1_A(\omega) \sup \Big\{x \in \mathbb{R}_+ \mid 1_A(\omega)  \int_0^{\infty}   \psi_x(\bP(-X > y\mid \sF_t)(\omega))   \dif y \\
		& \qquad+ 1_A(\omega)  \int^0_{-\infty}  [ \psi_x(\bP(-X > y\mid \sF_t)(\omega)) -1 ]  \dif y \leq 0 \Big\} \\
		&= 1_A(\omega) \sup \Big\{x \in \mathbb{R}_+ \mid \int_0^{\infty}   \psi_x(\bP(- 1_A X > y\mid \sF_t)(\omega))   \dif y \\
		& \qquad + \int^0_{-\infty}  [ \psi_x(\bP(- 1_A X > y\mid \sF_t)(\omega)) -1 ]  \dif y \leq 0 \Big\} \\
		&= 1_A(\omega) \alpha^{\Psi}_t( 1_A X)(\omega).
	\end{align*}	
	Hence locality of $\alpha_t^{\Psi}$ is established.

	\smallskip
	\noindent\textit{Quasi-concavity.}
	For any  $t \in \mathcal{T}$, $X, Y \in L^\infty$ and  $ \omega \in \Omega \backslash N^X$.  Let $m > 0 $ be some $\sF_{t}$-measurable random variable such that $\alpha^{\Psi}_{t}(X)(\omega) \geq m(\omega)$ and  $\alpha^{\Psi}_{t}\left(Y\right)(\omega) \geq m(\omega)$. \\
	Fix one such $\omega$ and denote $m^{\star} = m(\omega)$. By \eqref{eq:3-11}, we have that $\alpha^{\Psi}_{t}(X)(\omega) \geq m^{\star}$ and  $\alpha^{\Psi}_{t}\left(Y\right)(\omega) \geq m^{\star}$ imply that
	$$
	\rho_t^{\psi_{m(\omega)}}(X)(\omega)= \rho_t^{\psi_{m^{\star}}}(X)(\omega)  \leq 0,
	$$
	$$
	\rho_t^{\psi_{m(\omega)}}(Y)(\omega)= \rho_t^{\psi_{m^{\star}}}(Y)(\omega)  \leq 0,
	$$
	respectively. Since sublinearity implies convexity, then by  sublinearity of $\rho^{\psi}_{t}$,
	\begin{align} \label{eq:3.8}
		\rho_t^{\psi_{m^{\star}}}\left(\lambda X+ (1-\lambda) Y\right)(\omega) \leq \lambda \rho_t^{\psi_{m^{\star}}}( X)(\omega) + (1-\lambda) \rho_t^{\psi_{m^{\star}}}(Y)(\omega) \leq 0.
	\end{align}
	By \eqref{eq:3-11} again, we show \eqref{eq:3.8} is equivalent to $\alpha^{\Psi}_{t}\left(\lambda X+ (1-\lambda) Y\right)(\omega) \geq m^{\star}$.
	Therefore, quasi-concavity holds.
	
	\smallskip
	\noindent\textit{Monotonicity.}  For any $t \in \mathcal{T}$,  $X, Y \in L^\infty$ with $X \leq Y$ and $ \omega \in \Omega \backslash N^X$. For any $\sF_{t}$-measurable random variable $m> 0$, such that $\alpha_{t}^{\Psi}(X)(\omega) > m(\omega)$,   we have $  \rho_t^{\psi_{m(\omega)}}(X)(\omega) \leq 0$. By monotonicity of $\rho^{\psi}_{t}$, $  \rho_t^{\psi_{m(\omega)}}(Y)(\omega) \leq \rho_t^{\psi_{m(\omega)}}(X)(\omega) \leq 0$.
	Then,  $\alpha_{t}^{\Psi}(Y)(\omega) \geq m(\omega)$, and hence $\alpha_{t}^{\Psi}(Y)(\omega) \geq \alpha_{t}^{\Psi}(X)(\omega)$.

	\smallskip
	\noindent\textit{Scale Invariance.}
	For any $t \in \mathcal{T}$,  $X \in L^\infty$,  $\lambda \in L^{\infty}_t$ with  $\lambda> 0$, and $ \omega \in \Omega \backslash N^X$, we have that,
	$$
	\alpha_{t}^{\Psi}(\lambda X)(\omega) = \sup \left\{x \in \mathbb{R}_+ \mid \rho_t^{\psi_x}(\lambda X)(\omega) \leq 0 \right \}.
	$$
	In view of the positive homogeneity of $\rho^{\psi}_{t}$,
	$$
	\alpha_{t}^{\Psi}(\lambda X)(\omega) = \sup \left\{x \in \mathbb{R}_+ \mid \lambda \rho_t^{\psi_x}(X)(\omega) \leq 0 \right \}.
	$$
	Since $\lambda> 0$,  then $\lambda  \rho_t^{\psi_x}(X) \leq 0 $ is equivalent to $\rho_t^{\psi_x}(X) \leq 0 $. Hence, $\alpha_{t}^{\Psi}(\lambda X) = \alpha_{t}^{\Psi}(X)$.
	
	The proof is complete.
	
\end{proof}

\subsection{On conditional quantiles}

The conditional versions of qualtiles are defined naturally  using `probabilistic conditioning' of the corresponding regular (or static) notions. Same hold true for most properties. While morally this is true, the difficulties are hidden in technical details related to measurably and well-definiteness of these objects. Here we present some results related to that. We also refer to \cite{Bielecki2018} for some relevant discussion on this topic. 

\begin{definition}  \label{def:condquantile-supp} 
	For any $ \alpha \in (0,1)$, the conditional upper and lower $\alpha$-quantile of $X\in L^\infty$ with respect to $\sigma$-field $\sF_t$ are defined as,
	\begin{align*}
		q_{\alpha}^{+}(X \mid \sF_t) = \esssup \{m \in L^{\infty}_t \mid \mathbb{P}(X \leq m\mid\sF_t) \leq \alpha \},\\
		q_{\alpha}^{-}(X \mid \sF_t) =\essinf \{m \in L^{\infty}_t \mid  \mathbb{P}(X \leq m\mid\sF_t) \geq \alpha \}.
	\end{align*}
\end{definition}

Below wWe will give other representations of  the conditional quantile functions, and we start with an auxiliary result. 

\begin{lemma} \label{le:omega} 
	For any $X \in L^\infty, \ t \in \mathcal{T}, \ a \in L^{\infty}_t$,  and for fixed $\omega^{\prime} \in \Omega$, 
	\begin{align}
		P(X \leq a\mid\sF_t)(\omega^{\prime})= \bP(X \leq a(\omega^{\prime})\mid\sF_t)(\omega^{\prime}). 
	\end{align}
\end{lemma}
\begin{proof}
	For fixed $\omega^{\prime} \in \Omega$,
	\begin{align*}
		\bP(X \leq a\mid\sF_t)(\omega^{\prime})&= \1_{a=a(\omega^{\prime})} \bP(X \leq a\mid\sF_t)
		= \1_{a=a(\omega^{\prime})} \bE(\1_{X \leq a}\mid\sF_t) \\
		&= \1_{a=a(\omega^{\prime})} \bE(\1_{a=a(\omega^{\prime})} \1_{X \leq a}\mid\sF_t) 
		= \1_{a=a(\omega^{\prime})} \bE(\1_{X \leq a(\omega^{\prime})}\mid\sF_t) \\
		&= \bP(X \leq a(\omega^{\prime})\mid\sF_t)(\omega^{\prime}).
	\end{align*}
\end{proof}

\begin{lemma} \label{le:condquantile2} 
	For any $X \in L^\infty$, $t \in \mathcal{T}$, $\alpha \in (0,1)$ and $ \omega \in \Omega \backslash N^{X}$, the conditional lower and upper $\alpha$-quantile  admit the representation 
	\begin{align*}
		q_{\alpha}^{+}(X \mid \sF_t)(\omega)& = \sup \{x \in \mathbb{R} \mid \bP(X \leq x\mid\sF_t)(\omega) \leq \alpha \},  \\
		q_{\alpha}^{-}(X \mid \sF_t)(\omega)& = \inf \{x \in \mathbb{R} \mid \bP(X \leq x\mid\sF_t)(\omega) \geq \alpha \}. 
	\end{align*}
\end{lemma}
\begin{proof}	We will prove the first identity, and the second identity about $q_{\alpha}^{-}(X \mid \sF_t)$ can be proved similarly. For simplicity, denote by $m^{*}:= q_{\alpha}^{+}(X \mid \sF_t)$ and $x^{*}(\omega):= \sup \{x \in \mathbb{R} \mid \bP(X \leq x\mid\sF_t)(\omega) \leq \alpha \}$. We will show that the set $\{ m^{*} \neq x^{*}\}\cap (\Omega\setminus N^X)$ is an empty set. 
	We proceed by splitting   $\{ m^{*} \neq x^{*}\} $  into two sets $ \{ m^{*}< x^{*}\}$ and $ \{ m^{*}> x^{*}  \}$.
	
	First assume that $\{ m^{*}< x^{*}\} \cap(\Omega \setminus N^{X}) \neq \emptyset$. Then, there exist $M_1 \in \sF_t$ and $a \in L^{\infty}_t$, such that on non-empty set $M_1 \backslash N^{X}$,
	\begin{equation} \label{eq:le2.22-1}
		m^{*}<  a <x^{*}.
	\end{equation}
	Since $a <x^{*}$ on $ M_1 \backslash N^{X}$, we have for fixed $\omega^{\prime} \in M_1 \backslash N^{X}$, 
	\[
	a(\omega^{\prime}) <x^{*}(\omega^{\prime})= \sup \{x \in \mathbb{R} \mid \bP(X \leq x\mid\sF_t)(\omega^{\prime}) \leq \alpha \}.
	\]	
	Thus, $\bP(X \leq a(\omega^{\prime})\mid\sF_t)(\omega^{\prime}) \leq \alpha$, which by by Lemma~\ref{le:condquantile2} becomes, 
	$\bP(X \leq a\mid\sF_t)(\omega^{\prime}) \leq \alpha.$ Since $\omega^{\prime}$ can be chosen arbitrarily in $M_1 \backslash N^{X}$, we have, 
	$\bP(X \leq a\mid\sF_t)\1_{M_1 \backslash N^{X}} \leq \alpha.$
	By locality of conditional probability, $\bP(X \leq a\1_{M_1 \backslash N^{X}}\mid\sF_t) \leq \alpha,$ 	which by the definition of $q_{\alpha}^{+}(X \mid \sF_t)$ implies that $a\1_{M_1 \backslash N^{X}} \leq m^{*}$. This contradicts \eqref{eq:le2.22-1}, and thus  $\{ m^{*}< x^{*}\}=\emptyset$  on $\Omega \backslash N^{X}$. 
	
	Next we prove that $ \{ m^{*}> x^{*}  \} \cap (\Omega \backslash N^{X})$ is also an empty. Assume this is not true. Then, there exist  $M_2 \in \sF_t$ and $b \in L^{\infty}_t$, such that on non-empty set $M_2 \backslash N^{X}$,
	\begin{align} \label{eq:le2.22-2}
		x^{*}< b < m^{*}.
	\end{align}
	Since $ x^{*}< b $ on $ M_2 \backslash N^{X} $, for fixed $\omega^{\prime} \in M_2 \backslash N^{X}$, 
	\begin{align*}
		b(\omega^{\prime}) >x^{*}(\omega^{\prime})= \sup \{x \in \mathbb{R} \mid \bP(X \leq x\mid\sF_t)(\omega^{\prime}) \leq \alpha \}.
	\end{align*}
	Thus, $	\bP(X \leq b(\omega^{\prime})\mid\sF_t)(\omega^{\prime}) > \alpha$, and by Lemma~\ref{le:condquantile2}, $\bP(X \leq b\mid\sF_t)(\omega^{\prime}) > \alpha.$ Since $\omega^{\prime}$ can be chosen arbitrarily in $M_2 \backslash N^{X} $, we have on $M_2 \backslash N^{X} $, that $\bP(X \leq b\mid\sF_t)> \alpha$. By the definition of $q_{\alpha}^{+}(X \mid \sF_t)$,  we have that $b > m^{*}$ on $M_2 \backslash N^{X} $, which contradicts  \eqref{eq:le2.22-2}. 
	
	The proof is complete. 
\end{proof}

As a consequence of Lemma~\ref{le:condquantile2}, one can show that (see also \cite[Section 4.4]{FollmerSchiedBook2004} for unconditional case) 
\begin{align*}
	q_{\alpha}^{+}(X \mid \sF_t)(\omega)& = \sup \{x \in \mathbb{R} \mid \bP(X < x\mid\sF_t)(\omega) \leq \alpha \}. \\
	q_{\alpha}^{-}(X \mid \sF_t)(\omega)&= \inf \{x \in \mathbb{R} \mid \bP(X < x\mid\sF_t)(\omega) \geq \alpha \}.		
\end{align*}
For brefity, we show the first identity only. Let $a= \sup \{x \in \mathbb{R} \mid \bP(X < x\mid\sF_t)(\omega) \leq \alpha \}$ and $b= \sup \{x \in \mathbb{R} \mid \bP(X \leq x\mid\sF_t)(\omega) \leq \alpha \}$.  Then, since
$$
\left\{x \in \mathbb{R} \mid \bP(X < x\mid\sF_t)(\omega) \leq \alpha  \right\} \supset\left\{x \in \mathbb{R} \mid \bP(X \leq x\mid\sF_t)(\omega) \leq \alpha  \right\},
$$
we have $a \geq b$. To show that $a \leq b$, we proceed by contradiction. Assume that $a> b$, then there exists $s_{0}, s_{1} \in \mathbb{R}$, such that $b< s_{0}< s_{1}< a$. The inequality  $s_{1}< a$ implies that  $\bP(X < s_{1}\mid\sF_t)(\omega) \leq \alpha$, and $b< s_{0}$ implies that  $\bP(X \leq s_{0}\mid\sF_t)(\omega) > \alpha$. Consequently, we get
\begin{align} \label{re:condquan}
	\bP(X < s_{1}\mid\sF_t)(\omega) \leq \alpha < \bP(X \leq s_{0}\mid\sF_t)(\omega). 
\end{align}
On the other hand, by \cite[(A.4) and (A.5)]{Bielecki2018}, since $s_{0}< s_{1}$, we have that $\bP(X \leq s_{0}\mid\sF_t)(\omega) \leq \bP(X < s_{1}\mid\sF_t)(\omega) \leq \alpha$, which contradicts  \eqref{re:condquan}.

Finally, we present a result that relates lower conditional quantile to upper conditional quantile.
\begin{lemma} \label{le:2.20}
	For any $ X\in L^\infty, \ \omega \in \Omega \backslash N^{X}$,
	\begin{align*} 
		q^{+}_{\alpha}(X \mid \sF_t)(\omega)= - q^{-}_{1-\alpha}(-X \mid \sF_t)(\omega). 
	\end{align*} 
\end{lemma}
\begin{proof}
	By Lemma~\ref{le:condquantile2}, for any $ \omega \in \Omega \backslash N^{\prime}$,
	\begin{align*}
		q^{+}_{\alpha}(X \mid \sF_t)(\omega)
		&= \sup \{x \in \mathbb{R} \mid \bP(X < x\mid\sF_t)(\omega) \leq \alpha \}
		= - \inf \{-x \in \mathbb{R} \mid \bP(X < x\mid\sF_t)(\omega) \leq \alpha \}\\
		&= - \inf \{x \in \mathbb{R} \mid \bP(X < -x\mid\sF_t)(\omega) \leq \alpha \}\\
		& = - \inf \{x \in \mathbb{R} \mid \bP(-X > x\mid\sF_t)(\omega) \leq \alpha \}\\
		&= - \inf \{x \in \mathbb{R} \mid 1-\bP(-X \leq x\mid\sF_t)(\omega) \leq \alpha \}\\
		&= - \inf \{x \in \mathbb{R} \mid \bP(-X \leq x\mid\sF_t)(\omega) \geq 1- \alpha \}\\
		&= - q^{-}_{1-\alpha}(-X \mid \sF_t)(\omega). 
	\end{align*}
\end{proof}

\begin{lemma} \label{le:altcond}
	For any $X \in L^\infty$, $t \in \mathcal{T}$ and $ \omega \in \Omega \backslash N^X $,
	\begin{align*}
		\bE[X\mid\sF_t](\omega)=  \int_{(0,1)}   q^{-}_{1-z}(X \mid \sF_t)(\omega) \dif z.
	\end{align*}
\end{lemma}
\begin{proof}   By \cite[Theorem 5.4]{Kallenberg2006}, fix two measurable spaces  $S$  and  $T$, a $\sigma$-field  $\sF \subset \mathcal{A}$,  and a random element  $\xi$  in  $S$  such that  $\bP[\xi \in \cdot \mid \sF]$  has a regular version $ \nu$.  Further consider an  $\sF$-measurable random element  $\eta$  in  $T$  and  a  measurable function  $f$  on  $S \times T$  with  $\bE|f(\xi, \eta)|<\infty$. Then, for any $ \omega \in \Omega \backslash N^\xi $, 
	\begin{align} \label{eq:2.41}
		\bE[f(\xi, \eta) \mid \sF](\omega)=\int \nu(\dif  s)(\omega) f(s, \eta)(\omega).
	\end{align}
	Note that in our settings, there is only one variable $X$ which plays the role of $\xi$ in \eqref{eq:2.41}. Since $X \in L^\infty$, we only consider $X \geq 0$, otherwise take $X^{\prime}=X+C$, where $C:= \esssup X $. Thus we have, 
	\begin{align*}
		\bE[X\mid\sF_t](\omega)&=  \int_{\mathbb{R}} x  d\bP( X \leq x\mid\sF_t)(\omega) =\int_{[0,C]} x  \dif  \bP( X \leq x\mid\sF_t)(\omega)\\
		& =x \bP( X \leq x\mid\sF_t)(\omega)|^{C}_{0-}- \int_{[0,C]} \bP( X \leq  x\mid\sF_t)(\omega) \dif  x\\
		&= C- \int_{[0,C]} \left[ 1- \bP( X >  x\mid\sF_t)(\omega) \right] \dif x
		= \int_{[0,C]} \bP( X >  x\mid\sF_t)(\omega)  \dif x\\
		&= \int_{[0,C]} \int_{(0,1)}   \1_{ z <  \bP( X > x\mid\sF_t)(\omega)} \dif z \dif x
		= \int_{(0,1)}   \int_{[0,C]} \1_{ z <  \bP( X > x\mid\sF_t)(\omega)} \dif x \dif z\\
		&= \int_{(0,1)}   \int_{[0,C]} \1_{  \bP( X \leq x\mid\sF_t)(\omega) < 1-z  } \dif  x \dif z,
	\end{align*}
	Next, we will consider the continuous part and the discrete part of $X$ separately. For fixed $z \in (0,1)$, let $x^{\prime}= \inf \{x \in \mathbb{R} \mid \bP(X \leq x\mid\sF_t)(\omega) \geq 1- z \}$. 
	
	If $X$ is continuous at $x^{\prime}$, then $\bP( X = x^{\prime}\mid\sF_t)(\omega) =0$, we have $ \bP(X \leq x^{\prime}\mid\sF_t)(\omega) = 1- z $, and   
	\begin{align*}
		\int_{[0,C]} \1_{  \bP( X \leq x\mid\sF_t)(\omega) < 1-z  } \dif x &= C- \int_{[0,C]} \1_{  \bP( X \leq x\mid\sF_t)(\omega) \geq 1-z  } \dif x\\
		&= C- \int_{[x^{\prime}, C]} \1_{  \bP( X \leq x\mid\sF_t)(\omega) \geq 1-z  } \dif x\\
		&= C- \int_{[x^{\prime}, C]} \1 \dif x=  x^{\prime}. 
	\end{align*}
	
	If $X$ is discrete at $x^{\prime}$, then $\bP( X = x^{\prime}\mid\sF_t)(\omega) >0$. For those $z^{\prime}$ such that $ \bP(X \leq x^{\prime}-\mid\sF_t)(\omega) < 1- z^{\prime} \leq \bP(X \leq x^{\prime}\mid\sF_t)(\omega) $,  we deduce
	\begin{align*}
		\int_{[0,C]} \1_{  \bP( X \leq x\mid\sF_t)(\omega) < 1-z  } \dif x &= C- \int_{[0,C]} \1_{  \bP( X \leq x\mid\sF_t)(\omega) \geq 1-z  } \dif x\\
		&= C- \int_{[x^{\prime}, C]} \1_{  \bP( X \leq x\mid\sF_t)(\omega) \geq 1-z  } \dif x\\
		&= C- \int_{[x^{\prime}, C]} \1 dx=  x^{\prime}. 
	\end{align*}
	Thus for fixed $z \in (0,1)$, and any $X$, we have 
	\[
	\int_{[0,C]} \1_{  \bP( X \leq x\mid\sF_t)(\omega) < 1-z  } \dif x=  \inf \{x \in \mathbb{R} \mid \bP(X \leq x\mid\sF_t)(\omega) \geq 1- z \},
	\]
	and we continue  
	\begin{align*}
		\bE[X\mid\sF_t](\omega)  &= \int_{(0,1)}   \int_{[0,C]} \1_{  \bP( X \leq x\mid\sF_t)(\omega) < 1-z  } \dif x \dif z\\
		&=\int_{(0,1)}   \inf \{x \in \mathbb{R} \mid \bP(X \leq x\mid\sF_t)(\omega) \geq 1- z \} \dif z\\
		&= \int_{(0,1)}   q^{-}_{1-z}(X \mid \sF_t)(\omega) \dif z.
	\end{align*}
\end{proof}

\subsection{On $\var$ and  $\avar$}

The conditional $\var$,  similar to its static counterpart, is defined in terms of conditional $\alpha$-quantile function.

\begin{definition} For fixed $ \alpha \in (0,1)$ and  $X \in L^\infty$, conditional $\var$ at level $\alpha$ with respect to $\sigma$-field $\sF_t$  is defined as,
	\begin{align}  
		\var_{\alpha}(X \mid \sF_t) := -q_{\alpha}^{+}(X \mid \sF_t)= \essinf \{m \in L^{\infty}_t \mid  \mathbb{P}(X+ m < 0 \mid \sF_t) \leq \alpha \}.
	\end{align}	
\end{definition}
From financial point of view, $\var_{\alpha}(X|\sF_t)$ can be viewed as the smallest amount of capital, which, if added to the position $X$ at time $t$, will yield a secured position $X+m$ that encounters losses with a (conditional) probability below the level $\alpha$. Clearly, the $\var_\alpha$ does not capture the size or the distribution of the losses beyond the $\alpha$-quantile. To overcome this, the notion of average value at risk is introduced. 

\begin{definition} The conditional Average Value at Risk at level $\alpha \in (0, 1]$ of a position $ X \in L^\infty$ is given by
	\begin{align} \label{eq:avar-supp}
		\avar_{\alpha}(X\mid \sF_t):= \frac{1}{\alpha} \int_{(0, \alpha)} \var_{z}(X \mid \sF_t) d z= -\frac{1}{\alpha} \int_{(0, \alpha)} q_{z}^{+}(X \mid \sF_t) dz.
	\end{align} 
\end{definition}
Note that $ q_{z}^{+}(X \mid \sF_t)$ is monotone increasing with respect to $z$. Since monotonicity implies Borel measurability, $z \mapsto  q_{z}^{+}(X \mid \sF_t)$ is Borel measurable. Thus, the integral in \eqref{eq:avar-supp} is well-defined. 

In the existing literature, the $\avar$ sometimes is defined through the so-called robust or dual representations;  cf. \cite{Bielecki2018} for the conditional case.  Next results shows that these definitions are equivalent.

\begin{lemma}
	For any $ X \in L^\infty$ and $\alpha \in (0,1]$,
	\begin{equation} \label{eq:avarrobust}
		\avar_{\alpha}(X\mid\sF_t)= \esssup \{ \bE[-XZ\mid\sF_t]\mid Z \in \sF, 0 \leq Z \leq 1/\alpha, \bE(Z\mid\sF_t)=1 \}.
	\end{equation} 
	Moreover, a maximizer  $Z^{*}$  in the right hand side of \eqref{eq:avarrobust}  exists, and it is given by 
	\begin{equation} \label{eq:rnavar}
		Z_{\alpha}^{*}=\frac{1}{\alpha}\left(\1_{X<q_{\alpha}^{\pm}(X \mid \sF_t)}+ \varepsilon \1_{X=q_{\alpha}^{\pm}(X \mid \sF_t)}\right),
	\end{equation}
	here
	\begin{align*}
		\varepsilon=\left\{\begin{array}{ll}
			0, & \bP(X=q_{\alpha}^{\pm}(X \mid \sF_t) \mid \sF_t)=0, \\
			\frac{\alpha- \bP(X<q_{\alpha}^{\pm}(X \mid \sF_t) \mid \sF_t)}{\bP(X=q_{\alpha}^{\pm}(X \mid \sF_t) \mid \sF_t)}, & \text {\  otherwise}.
		\end{array}\right.
	\end{align*}
\end{lemma}

\begin{proof}

	In view of \eqref{eq:avar-supp}, it is enough to show that, 
	\begin{equation}  \label{eq:altavar1}
		-\frac{1}{\alpha}  \int_{(0,\alpha)}  q^{+}_{z}(X \mid \sF_t)(\omega)  \dif z= \bE[-XZ_{\alpha}^*\mid\sF_t](\omega).
	\end{equation}
	Let $L= -\frac{1}{\alpha}  \int_{(0,\alpha)}  q^{+}_{z}(X \mid \sF_t)(\omega)  \dif z$. Then,  
	\begin{equation}  \label{eq:altavar2}
		\begin{split}	
			L &= \frac{1}{\alpha} \left[ \int_{(0,\alpha)}  q^{+}_{\alpha}(X \mid \sF_t)(\omega)- q^{+}_{z}(X \mid \sF_t)(\omega) \dif z  \right] -q^{+}_{\alpha}(X \mid \sF_t)(\omega) \\
			&= \frac{1}{\alpha} \left[ \int_{(0,1)} \left(  q^{+}_{\alpha}\left(X \mid \sF_t\right)(\omega)- q^{+}_{z}\left(X \mid \sF_t\right)(\omega)\right)^+  \dif z  \right] -q^{+}_{\alpha}(X \mid \sF_t)(\omega).
		\end{split}
	\end{equation}
	
	Next we prove an intermediary equality
	\begin{equation} \label{eq:altavar3}
		\bE\left[ \left(q^{+}_{\alpha}\left(X \mid \mathcal{F}_t\right) -X \right)^+ \mid\mathcal{F}_t\right](\omega)
		=  \int_{(0,1)}  \left(  q^{+}_{\alpha}\left(X \mid \mathcal{F}_t\right)(\omega)- q^{+}_{z}\left(X \mid \mathcal{F}_t\right)(\omega) \right)^+  \dif z.
	\end{equation}
	By Lemma~\ref{le:altcond}
	\begin{equation} \label{eq:altavar4}
		\begin{split}
			\bE\big[ \left(q^{+}_{\alpha}(X \mid \mathcal{F}_t) -X \right)^+ &\mid\mathcal{F}_t\big](\omega) 
			= \int_{(0,1)}   q^{-}_{1-z}\left( \left(q^{+}_{\alpha}(X \mid \mathcal{F}_t) -X \right)^+ \mid \mathcal{F}_t\right)(\omega) \dif z \\
			&  = \int_{(0,1)} \sup \{x \in \mathbb{R} \mid \bP\left( (q^{+}_{\alpha}(X \mid \mathcal{F}_t) -X)^+ < x\mid\mathcal{F}_t \right)(\omega) < 1- z \}  \dif z.
		\end{split}
	\end{equation}
	Since 
	\[
	\bP(X^+ < x)= \bP(\max\{ 0, X\} < x )= \bP(0 < x, X < x )= 
	\begin{cases}
		0, &x \geq 0,\\
		\bP(X< x ), &x > 0, 
	\end{cases}
	\]
	we have that 
	\begin{align} \label{eq:altavar5}
		&\sup \{x \in \mathbb{R} \mid \bP\left( \left( q^{+}_{\alpha}(X \mid \mathcal{F}_t) -X \right)^+ < x\mid\mathcal{F}_t \right)(\omega) < 1- z \} \\
		& \qquad \qquad =  \sup \{x > 0 \mid \bP\left( q^{+}_{\alpha}(X \mid \mathcal{F}_t) -X < x\mid\mathcal{F}_t \right)(\omega) < 1- z \}.
	\end{align}
	Then, we combine \eqref{eq:altavar4} with \eqref{eq:altavar5}, and obtain 
	\[ 
	\bE\left[ \left(q^{+}_{\alpha}(X \mid \mathcal{F}_t) -X \right)^+ \mid\mathcal{F}_t\right](\omega) = 
	\int_{(0,1)} \sup \{x > 0 \mid \bP\left( q^{+}_{\alpha}(X \mid \mathcal{F}_t) -X < x\mid\mathcal{F}_t \right)(\omega) < 1- z \} \dif z.
	\] 
	Using this, to show \eqref{eq:altavar3},  it is sufficient to show that for any $z \in (0, 1)$,
	\begin{equation} \label{eq:altavar7}
		\sup \{x > 0 \mid \bP\left( q^{+}_{\alpha}(X \mid \mathcal{F}_t) -X < x\mid\mathcal{F}_t \right)(\omega) < 1- z \} 
		\left(  q^{+}_{\alpha}(X \mid \mathcal{F}_t)(\omega)- q^{+}_{z}(X \mid \mathcal{F}_t)(\omega) \right)^+, 
	\end{equation}
	which we prove next. Denote by $J:= \sup \{x > 0 \mid \bP\left( q^{+}_{\alpha}(X \mid \mathcal{F}_t) -X < x\mid\mathcal{F}_t \right)(\omega) < 1- z \}$. Then, 
	\begin{align*}
		J&= \sup \{x > 0 \mid \bP\left(  -X < -q^{+}_{\alpha}(X \mid \mathcal{F}_t)(\omega) + x  \mid\mathcal{F}_t \right)(\omega) < 1- z \}\\
		&= \sup \{x > 0 \mid \bP\left(  X > q^{+}_{\alpha}(X \mid \mathcal{F}_t)(\omega) - x  \mid\mathcal{F}_t \right)(\omega) < 1- z \}\\
		&= \sup \{x > 0 \mid \bP\left(  X \leq q^{+}_{\alpha}(X \mid \mathcal{F}_t)(\omega) - x  \mid\mathcal{F}_t \right)(\omega) >  z \}.
	\end{align*}
	Let $y:= q^{+}_{\alpha}(X \mid \mathcal{F}_t)(\omega) - x$, and we continue 
	\begin{align*}
		J &= \sup \{q^{+}_{\alpha}(X \mid \mathcal{F}_t)(\omega) -y > 0 \mid \bP\left(  X \leq y  \mid\mathcal{F}_t \right)(\omega) >  z \}\\
		&=q^{+}_{\alpha}(X \mid \mathcal{F}_t)(\omega)+ \sup \{ -y >- q^{+}_{\alpha}(X \mid \mathcal{F}_t)(\omega) \mid \bP\left(  X \leq y  \mid\mathcal{F}_t \right)(\omega) >  z \}\\
		&=q^{+}_{\alpha}(X \mid \mathcal{F}_t)(\omega)- \inf \{ y < q^{+}_{\alpha}(X \mid \mathcal{F}_t)(\omega) \mid \bP\left(  X \leq y  \mid\mathcal{F}_t \right)(\omega) >  z \}\\
		&=  q^{+}_{\alpha}(X \mid \mathcal{F}_t)(\omega)- \min \left\{ q^{+}_{\alpha}(X \mid \mathcal{F}_t)(\omega), \inf \left\{ y \in \mathbb{R} \mid \bP\left(  X \leq y  \mid\mathcal{F}_t \right)(\omega) >  z \right \} \right \}\\
		&=  q^{+}_{\alpha}(X \mid \mathcal{F}_t)(\omega)+ \max \left\{ -q^{+}_{\alpha}(X \mid \mathcal{F}_t)(\omega), -\inf \left \{ y \in \mathbb{R} \mid \bP\left(  X \leq y  \mid\mathcal{F}_t \right)(\omega) >  z \right \} \right \}\\
		&=  \max \left\{ 0, q^{+}_{\alpha}(X \mid \mathcal{F}_t)(\omega)- q^{+}_{z}(X \mid \mathcal{F}_t)(\omega) \right \}\\
		&=  \left(q^{+}_{\alpha}(X \mid \mathcal{F}_t)(\omega)- q^{+}_{z}(X \mid \mathcal{F}_t)(\omega) \right)^+,
	\end{align*}
	and thus \eqref{eq:altavar7}, is proved, and consequently \eqref{eq:altavar3} is established.

	Then by \eqref{eq:altavar2} and \eqref{eq:altavar3}, we get,
	\begin{align*}
		L &=  \frac{1}{\alpha}\bE\left[ \left(q^{+}_{\alpha}(X \mid \mathcal{F}_t) -X \right)^+ \mid\mathcal{F}_t\right](\omega) -q^{+}_{\alpha}(X \mid \mathcal{F}_t)(\omega)\\
		&=  \frac{1}{\alpha}\bE\left[ \left(q^{+}_{\alpha}(X \mid \mathcal{F}_t) -X \right) \1_{X<q^{+}_{\alpha}(X \mid \mathcal{F}_t)} \mid\mathcal{F}_t\right](\omega) -q^{+}_{\alpha}(X \mid \mathcal{F}_t)(\omega)\\
		&= \frac{1}{\alpha}q^{+}_{\alpha}\left(X \mid \mathcal{F}_t\right)(\omega) \bP\left(X<q_{\alpha}^{+}(X \mid \mathcal{F}_t) \mid \mathcal{F}_t\right)(\omega) \\
		& \qquad \qquad + \frac{1}{\alpha} \bE\left[-X \1_{X<q^{+}_{\alpha}(X \mid \mathcal{F}_t)} \mid\mathcal{F}_t\right](\omega) -q^{+}_{\alpha}(X \mid \mathcal{F}_t)(\omega) \\
		&= -\frac{1}{\alpha}q^{+}_{\alpha}\left(X \mid \mathcal{F}_t\right)(\omega) \left[  \alpha- \bP\left(X<q_{\alpha}^{+}(X \mid \mathcal{F}_t) \mid \mathcal{F}_t\right)(\omega) \right] \\
		& \qquad \qquad + \frac{1}{\alpha} \bE\left[-X \1_{X<q^{+}_{\alpha}(X \mid \mathcal{F}_t)} \mid\mathcal{F}_t\right](\omega) \\
		&= - \frac{1}{\alpha}q_{\alpha}^{+}(X \mid \mathcal{F}_t)(\omega)  \bE\left[ \frac{\alpha-\bP(X<q_{\alpha}^{+}(X \mid \mathcal{F}_t) \mid \mathcal{F}_t)}{\bP(X=q_{\alpha}^{+}(X \mid \mathcal{F}_t) \mid \mathcal{F}_t)} \1_{X=q^{+}_{\alpha}(X \mid \mathcal{F}_t)} \mid\mathcal{F}_t\right](\omega) \\
		& \qquad \qquad +  \frac{1}{\alpha} \bE\left[-X \1_{X<q^{+}_{\alpha}(X \mid \mathcal{F}_t)} \mid\mathcal{F}_t\right](\omega)\\
		&=  \frac{1}{\alpha} \bE\left[ -q_{\alpha}^{+}(X \mid \mathcal{F}_t) \frac{\alpha-\bP(X<q_{\alpha}^{+}(X \mid \mathcal{F}_t) \mid \mathcal{F}_t)}{\bP(X=q_{\alpha}^{+}(X \mid \mathcal{F}_t) \mid \mathcal{F}_t)} \1_{X=q^{+}_{\alpha}(X \mid \mathcal{F}_t)} \mid\mathcal{F}_t\right](\omega)  \\
		& \qquad \qquad +  \frac{1}{\alpha} \bE\left[-X \1_{X<q^{+}_{\alpha}(X \mid \mathcal{F}_t)} \mid\mathcal{F}_t\right](\omega)\\
		&=  \frac{1}{\alpha} \bE\left[ -X \frac{\alpha-\bP(X<q_{\alpha}^{+}(X \mid \mathcal{F}_t) \mid \mathcal{F}_t)}{\bP(X=q_{\alpha}^{+}(X \mid \mathcal{F}_t) \mid \mathcal{F}_t)} \1_{X=q^{+}_{\alpha}(X \mid \mathcal{F}_t)} \mid\mathcal{F}_t\right](\omega) \\
		& \qquad \qquad +  \frac{1}{\alpha} \bE\left[-X \1_{X<q^{+}_{\alpha}(X \mid \mathcal{F}_t)} \mid\mathcal{F}_t\right](\omega)\\
		&=  \frac{1}{\alpha} \bE\left[ -X \varepsilon \1_{X=q^{+}_{\alpha}(X \mid \mathcal{F}_t)} \mid\mathcal{F}_t\right](\omega)+ \frac{1}{\alpha} \bE\left[-X \1_{X<q^{+}_{\alpha}(X \mid \mathcal{F}_t)} \mid\mathcal{F}_t\right](\omega)\\
		&= \bE[-XZ_{\alpha}^*\mid\mathcal{F}_t](\omega),
	\end{align*}
	where $Z_{\alpha}^*$ is given as \eqref{eq:rnavar}. This is exactly \eqref{eq:altavar1}, and the proof is complete. 
\end{proof}

\end{document}